\documentclass[%
 aip,
 jmp,%
 amsmath,amssymb,
reprint,%
author-year,%
]{revtex4-2}

\usepackage{graphicx}
\usepackage{dcolumn}
\usepackage{bm}
\usepackage{subcaption}
\usepackage{amsthm}
\usepackage{booktabs}
\usepackage{makecell}
\usepackage{multirow}
\usepackage{array}
\usepackage{mathtools}
\newtheorem{theorem}{Theorem}
\DeclarePairedDelimiter{\abs}{\lvert}{\rvert}
\newcommand{\R}{\mathbb{R}}   
\newcolumntype{C}[1]{>{\centering\arraybackslash}p{#1}}

\begin{document}

\preprint{AIP/123-QED}

\title[]{Stochastic modeling of stratospheric temperature}
\thanks{The PhD grant of M.\ D.\ Eggen is funded by NORSAR. This work was supported by the Research Council of Norway, FRIPRO Young Research Talent SCROLLER project, grant number 299897, as well as the Research Council of Norway FRIPRO/FRINATEK project MADEIRA, grant number 274377.}

\author{Mari Dahl Eggen}
 \email{marideg@math.uio.no.}
\affiliation{ 
Department of Mathematics, University of Oslo, P.O.\ Box 1053 Blindern, 0316 Oslo, Norway.
}

\author{Kristina Rognlien Dahl}
\affiliation{ 
Department of Mathematics, University of Oslo, P.O.\ Box 1053 Blindern, 0316 Oslo, Norway.
}

\author{Sven Peter N{\"a}sholm}
\affiliation{
NORSAR, Gunnar Randers vei 15, 2027 Kjeller, Norway
}
\affiliation{
Department of Informatics, University of Oslo, P.O.\ Box 1080 Blindern, 0316 Oslo, Norway
}

\author{Steffen M{\ae}land}
\affiliation{
NORSAR, Gunnar Randers vei 15, 2027 Kjeller, Norway
}

\date{\today}

\begin{abstract}
This study suggests a stochastic model for time series of daily-zonal (circumpolar) mean stratospheric temperature at a given pressure level. %
It can be seen as an extension of previous studies which have developed stochastic models for surface temperatures. %
The proposed model is a sum of a deterministic seasonality function and a L{\'e}vy-driven multidimensional Ornstein-Uhlenbeck process, which is a mean-reverting stochastic process. %
More specifically, the deseasonalized temperature model is an order $4$ continuous time autoregressive model, meaning that the stratospheric temperature is modeled to be directly dependent on the temperature over four preceding days, while the model's longer-range memory stems from its recursive nature. %
%
This study is based on temperature data from the European Centre for Medium-Range Weather Forecasts ERA-Interim reanalysis model product. 
The residuals of the autoregressive model are well-represented by normal inverse Gaussian distributed random variables scaled with a time-dependent volatility function. %
A monthly variability in speed of mean reversion of stratospheric temperature is found, hence suggesting a generalization of the $4$th order continuous time autoregressive model. %
A stochastic stratospheric temperature model, as proposed in this paper, can be used in geophysical analyses to improve the understanding of stratospheric dynamics. In particular, such characterizations of stratospheric temperature may be a step towards greater insight in modeling and prediction of large-scale middle atmospheric events, such as for example sudden stratospheric warmings. %
Through stratosphere-troposphere coupling, the stratosphere is hence a source of extended tropospheric predictability at weekly to monthly timescales, which is of great importance in several societal and industry sectors. The stochastic model might for example contribute to improved pricing of temperature derivatives.  
\end{abstract}

\maketitle


\section{\label{sec: introduction}Introduction}
A thorough understanding of surface weather dynamics is crucial in a wide range of industry and societal sectors. 
Whether planning marine operations, flights or farming, or managing energy assets, the weather is a key aspect to consider. However, because higher atmospheric layers can couple to levels closer to the surface, in order to understand weather, understanding the dynamics at higher altitudes of the atmosphere is important. The Earth's atmosphere has a layered structure, where each layer has layer-specific properties (see \cite{baldwin2019_100years} and the references therein for an historical overview). 
Closest to the surface lays the troposphere, reaching up to around $15$\;km altitude. 
Above, up to around $50$\;km, lays the stratosphere, which is the atmospheric layer of interest in this paper. 
These two layers interact, and the dynamics in the stratosphere can couple to the troposphere to affect dynamics and predictability at the surface, see for example \cite{butler2019chapter} and \cite{BaldwinDunkerton2001}. Hence, better probing, modeling, and understanding of stratospheric dynamics has the potential to enhance numerical surface weather prediction, in particular at weekly to monthly timescales.
%
%
%

In the current paper, a novel stochastic model for stratospheric temperature is proposed. 
%
%
The stochastic approach is similar to what has been applied in previous tropospheric temperature and wind dynamics modeling studies (e.g.,  \cite{benth05}, \cite{benth2008}, \cite{benth09}, \cite{benth13} and \cite{benth14}). 

Temperature tends to revert back to its mean over time (see, e.g.,  \cite{benth2008}). This feature is reflected in what is called the speed of mean reversion, and is captured by autoregressive (AR) processes. AR processes are discrete time stochastic processes having a direct transformation relation with continuous time autoregressive (CAR) processes \cite{brockwell04}, \cite{benth2008}, \cite{benth13}. 
This transformation allows to introduce a continuous time mathematical model framework based on empirical derivations and analyses. 
%
Periodical behaviour 
is modeled separately from the CAR process. So is a long-term trend in the stratospheric temperature. The reason for the inclusion of a long-term trend in stochastic models for tropospheric temperature, is that it is well known from climate research that there is a long-term warming of the troposphere (see, e.g., \cite{jones1990globalwarming}, \cite{sevellec2018globalwarming}, \cite{climateChange2013}). This effect is captured in the surface temperature modeling of \cite{benth2008}. Also in \cite{benth2008}, cyclic functions are included through truncated Fourier series which can represent the periodical movements of tropospheric temperature. 
Similarly, several studies have shown (e.g.,  \cite{cnossen2015}, \cite{danilov20}, \cite{steiner20}), that there is a long-term decreasing trend in stratospheric temperature, and that there are several cyclic (seasonal) patterns \cite{fu10}, \cite{mccormack96}. 
%
%

The current empirical analysis and stochastic model study is performed on temperatures as represented in the European Centre for Medium-Range Weather Forecasts (ECMWF) ERA-Interim atmospheric reanalysis model product (see \cite{ecmwf_stratospheric_temp} and \cite{dee2011}). Full-year daily-zonal mean stratospheric temperature data over $60^{\circ}$N and $10$\;hPa, from 1979 to 2018, are considered. Similarly as in \cite{benth2008}, a seasonality function in the form of a truncated Fourier series is fitted to the stratospheric temperature data to find deseasonalized temperature. Then, an AR model is fitted to the deseasonalized data and thereby subtracted to find the residuals. A search for seasonal heteroskedasticity (variability of variance over time) in the residuals is performed, and such heteroskedasticity is found. Based on this, a time-dependent volatility function, $\sigma(t)$, is defined. The $\sigma(t)$-scaled residuals are proven, with statistical significance, to be normal inverse Gaussian (NIG) distributed random variables. 
Hence, the results of the data analysis suggest using a L{\'e}vy-driven CAR process with a time-varying volatility function to model stratospheric temperature. The residuals contain small memory effects, indicating that it might be reasonable to also consider a stochastic volatility function. This is beyond the scope of this paper. 
%
%

An individual stability analysis of speed of mean reversion over time is also performed, suggesting that the assumption about constant speed of mean reversion is not fulfilled. 
The result is twofold; the speed of mean reversion shows large variability from month to month, and it is varying with a seasonal pattern. 
Similarly, \cite{zapranis08} proved that the speed of mean reversion for tropospheric temperature is strongly time-dependent, obtaining a series of daily values of speed of mean reversion through neural networks. However, in contrast to the current paper, they did not observe seasonal patterns. 
In \cite{benth15}, Ornstein-Uhlenbeck (OU) dynamics are generalized to allow for a stochastic speed of mean reversion, which can incorporate deterministic time dependence as well. However, \cite{benth15} consider the special case when the L{\'e}vy process is a Brownian motion. This is less general than the CAR process proposed for stratospheric temperature modeling in the current paper. Instead of using the aforementioned approaches to include time dependence in the speed of mean reversion of the CAR process, a simpler, approximate approach is suggested: A time-dependent step function with $12$ levels is introduced, where the levels represent the months of the year. In this way both the monthly variability and the seasonal behaviour are adjusted for. The procedure of fitting a CAR process to the stratospheric temperature data is repeated for the extended CAR process with time dependence in speed of mean reversion. The inclusion of time dependence does not change the outcome of $\sigma(t)$-scaled residuals: These are still NIG distributed random variables. 
%
%

The structure of the paper is as follows: In Sect.\,\ref{sect:preliminaries_ch}, a mathematical framework of the stochastic model for stratospheric temperature is proposed. In Sect.\,\ref{sect:main_analysis_ch}, a non-Gaussian CAR process with constant speed of mean reversion and time-dependent volatility function is introduced and proposed to model stratospheric temperature. The methodology for fitting the model to the stratospheric temperature data is explained. Furthermore, it is shown that the most important features of stratospheric temperature dynamics are explainable through the proposed model. Then, in Sect.\,\ref{sec:mean_reversion} a stability analysis of speed of mean reversion is performed, revealing that the proposed CAR process should be generalized to include a time-dependent speed of mean reversion. Finally, conclusions and suggestions for future work are provided in Sect.\,\ref{sec: conclusions}.


\section{\label{sect:preliminaries_ch}The structure of a stratospheric temperature model}
In this section, it is argued that stratospheric temperature exhibits an autoregressive behavior. This motivates the use of autoregressive models. For completeness, the definitions of discrete and continuous time autoregressive processes are recalled. It is also explained how these two kinds of processes are connected. 


\subsection{Autoregressive behaviour of stratospheric temperature dynamics\label{sec:autoreg_behav}}
%
%
%
Inspection of daily-zonal mean stratospheric temperature data over $60^{\circ}$N and $10$\;hPa, from now on referred to as the stratospheric temperature data, $S(t)$, clearly indicates a seasonal pattern. This is illustrated in Fig.\,\ref{fig:temp_10y_with_seasonfit} which displays 10 years of stratospheric temperature data. 
The corresponding autocorrelation function (ACF), computed over 1 January 1979 to 31 December 2018 with lags up to $730$ days ($2$ years), is presented in Fig.\,\ref{fig:acf_strat_temp}. The ACF pattern confirms a stratospheric temperature seasonal behaviour. %
\par
\begin{figure}[ht!]
    \centering
    \includegraphics[scale=0.5]{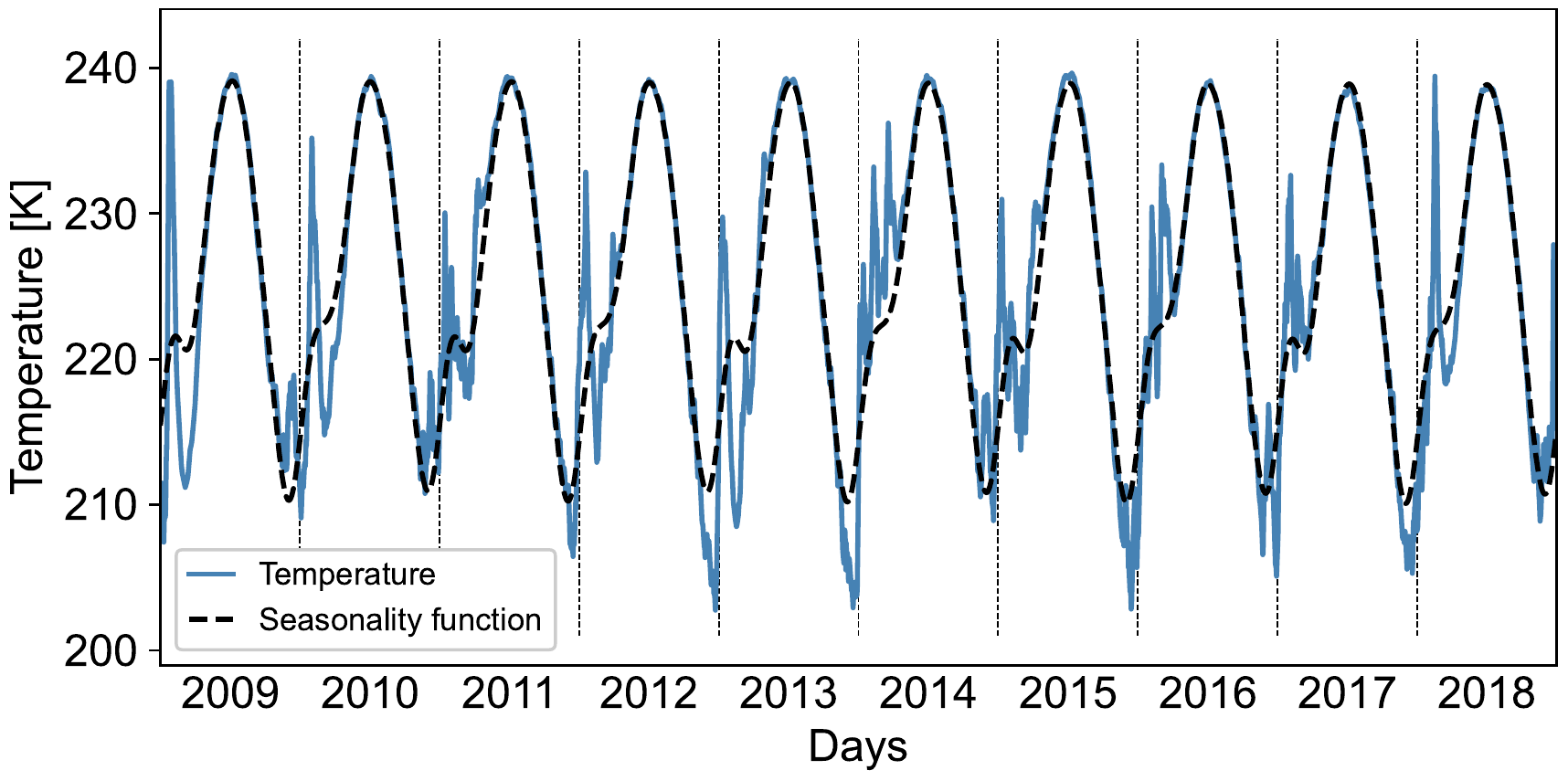}
    \caption{Daily-zonal mean stratospheric temperature, $S(t)$, over a region $R$ (see Sect.\,\ref{sec:strat_temp_data_retrieval}) the last 10 years (1 January 2009 to 31 December 2018) with the fitted seasonality function $\Lambda(t)$. The vertical lines represent each of the 10 years}
    \label{fig:temp_10y_with_seasonfit}
\end{figure} 
\begin{figure}[ht!]
	\centering
	\begin{subfigure}[ht!]{0.3\textwidth}
	\centering	
	\includegraphics[scale=0.6]{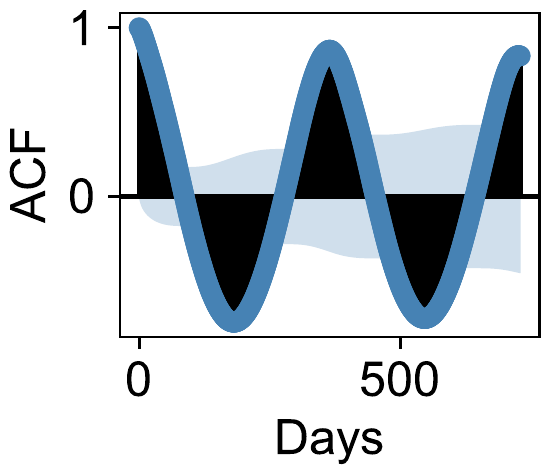}
	\caption{ACF of $S(t)$}
	\label{fig:acf_strat_temp}
	\end{subfigure}
	\begin{subfigure}[ht!]{0.3\textwidth}
	\centering	
	\includegraphics[scale=0.6]{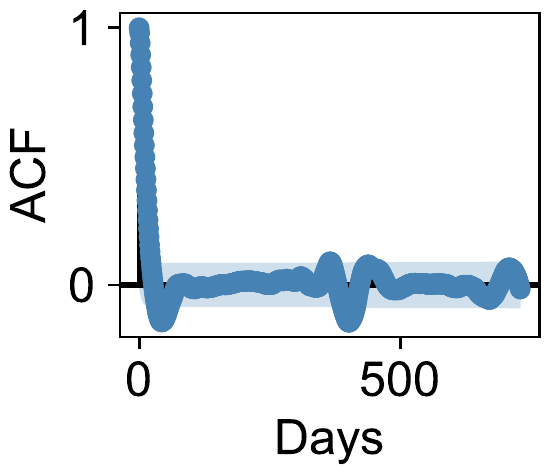}
	\caption{ACF of $Y(t)$}
	\label{fig:acf_strat_temp_ds}
	\end{subfigure}
		\begin{subfigure}[ht!]{0.3\textwidth}
	\centering	
	\includegraphics[scale=0.6]{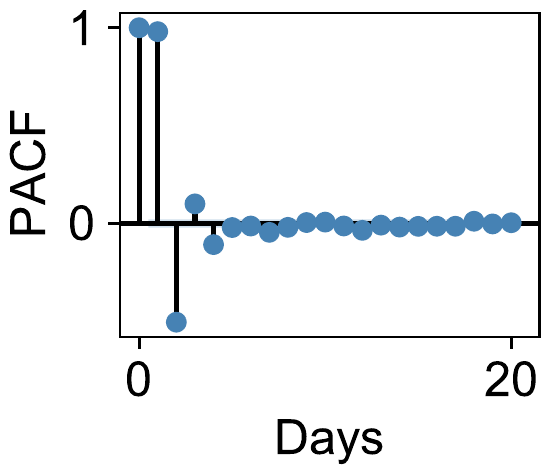}
	\caption{PACF of $Y(t)$}
	\label{fig:pacf_strat_temp_ds}
	\end{subfigure}
    \caption{ACF of stratospheric temperature, $S(t)$, and ACF and PACF of deseasonalized stratospheric temperature, $Y(t)$ (see Eq.\,\eqref{eq:general_temp_model})}
    \label{fig:acf_pacf_temp_ds}
\end{figure} 
Inspired by prior work in stochastic modeling of surface temperature and wind dynamics, in the context of financial weather contracts \citep[see, e.g.,][]{benth05,benth2008,benth09,benth13,benthbook13}, a long-term seasonality and trend function is fit to the stratospheric temperature data, see Sect.\,\ref{sec:season_function_data}. The long-term seasonality and trend function is from now on referred to as the seasonality function and deseasonalized temperature is obtained by subtracting the fitted seasonality function from the original dataset. 

Denote by $Y(t)$ the deseasonalized version of the stratospheric temperature, $S(t)$. Further, define $\Lambda (t): [0,T]\to \mathbb{R}$ to be a bounded and continuously differentiable (deterministic) seasonality function. Thus, stratospheric temperature is modeled as
\begin{align}
\label{eq:general_temp_model}
    S(t) = \Lambda(t) + Y(t).
\end{align}
Note that although the seasonality function $\Lambda(t)$ is deterministic, the stratospheric temperature, $S(t)$, and the deseasonalized temperature, $Y(t)$, are stochastic. Let $\Omega$ be a scenario space. Then, both $S(t)$ and $Y(t)$ depend on some scenario $\omega \in \Omega$, that is, $S(t) \triangleq S(t,\omega)$, $Y(t) \triangleq Y(t,\omega)$. For notational convenience, the scenario $\omega$ is suppressed from the notation for the remaining part of the paper.
In Sect.\,\ref{sec:season_function_data}, a review of possible seasonal effects is given prior to the explicit definition of the seasonality function $\Lambda(t)$. There, it will also be shown that a truncated Fourier series with linear trend is an appropriate choice for the seasonality function (see Eq.\,\eqref{eq:lambda_def}). 

Studying the ACF and partial autocorrelation function (PACF) of the deseasonalized temperature data, $Y(t)$, (Figs.\,\ref{fig:acf_strat_temp_ds} and \ref{fig:pacf_strat_temp_ds}) it is found that the deseasonalized temperature dynamics follows an AR process. For completeness, the definition of AR processes is given in the next section. Further, the PACF of the deseasonalized stratospheric temperature (Fig.\,\ref{fig:pacf_strat_temp_ds}) indicates that an AR($4$) model should be used to capture significant memory effects (see \cite{Levendis18} for an introduction to AR modeling and the interpretation of PACF plots). 
This means that the direct memory effects in stratospheric temperature last for four days in this model. However, due to the recursive properties of AR processes, the total memory effect is actually longer. 
As explained in Sect.\,\ref{sect:car_ar_connection}, there is a transformation relation between (discrete time) AR($p$) and (continuous time) CAR($p$) processes. This means that, by removing the seasonal behaviour in stratospheric temperature data, the resulting deseasonalized stratospheric temperature data can be modeled by a CAR process. It will be shown that $Y(t)$ can be approximated by a CAR($4$) model. 


\subsection{Non-Gaussian CAR($p$) processes and their connection to AR($p$) processes\label{sect:car_ar_connection}}

In this section, AR and CAR processes are defined. Motivated by the observed ACF and PACF for the deseasonalized stratospheric temperature data (see Sect.\,\ref{sec:autoreg_behav}), these processes are natural components in a model for stratospheric temperature. They are also part of surface temperature models applied in energy markets contexts (\cite{benth05}, \cite{benth13}, \cite{benth2008}, Ch.\,10 and \cite{benthbook13}, Ch.\,4). 

Suppose $(\Omega, \mathcal{F}, \{\mathcal{F}_t\}_{t \geq 0}, P)$ is a complete filtered probability space. Let $\boldsymbol{X}(t)=\{\boldsymbol{X}(t)\}_{t\in\mathbb{R}^+}$ be a stochastic process in $\mathbb{R}^p$, $p \in \mathbb{N}$, defined by a multidimensional non-Gaussian OU process with time-dependent volatility. 
That is, $\boldsymbol{X}(t)$ is given by the solution of the stochastic differential equation (SDE)
\begin{align}
\label{eq:car(p)-model}
    d\boldsymbol{X}(t)= A\boldsymbol{X}(t)dt + \boldsymbol{e}_p\sigma(t-)dL(t),
\end{align}
where $\boldsymbol{e}_p$ is the $p$-th unit vector in $\mathbb{R}^p$, $\sigma(t): \mathbb{R}^{+}\to \mathbb{R}^{+}$ is a càdlàg, $\mathcal{F}_t$-adapted function, $L(t) = \{L(t)\}_{t\in\mathbb{R}^+}$ is a L{\'e}vy process, and $A$ is the $p\times p$ coefficient matrix
\begin{align}
    \label{eq:ou_matrix_p}
    A = \begin{bmatrix}
       0 & 1 & 0 & \cdots & 0       \\
       0 & 0 & 1 &\cdots & 0        \\
       \vdots & \vdots & \vdots &\vdots &\vdots \\
       0 & 0 & 0 & \cdots & 1 \\
       -\alpha_p & -\alpha_{p-1} & -\alpha_{p-2} & \cdots & -\alpha_1 \\
     \end{bmatrix},
\end{align}
where $\alpha_k$, $k=1,\ldots , p$, are positive constants. By the multidimensional Itô formula the solution of Eq.\,\eqref{eq:car(p)-model} is given by
\begin{align*}
    \boldsymbol{X}(s) = \exp\left( A(s-t) \right)\boldsymbol{x} + \int_{t}^{s}\exp(A(s-u))\boldsymbol{e}_p\sigma(u-)dL(u),
\end{align*}
where $s\geq t \geq 0$ and $\boldsymbol{X}(t) \triangleq \boldsymbol{x}\in\R^p$. Let $p>j>q\in\mathbb{N}$. Then, the L{\'e}vy-driven stochastic process $Y(t) = \{Y(t)\}_{t\in\mathbb{R^+}}$ defined as $Y(t) \triangleq \boldsymbol{b}'\boldsymbol{X}(t)$, for a transposed vector $\boldsymbol{b}\in\mathbb{R}^{p\times 1}$ with elements satisfying $b_q\neq 0$ (sometimes, $b_q=1$ is assumed) and $b_j=0$, is called a (L{\'e}vy-driven) CARMA($p,q$) process (e.g.,  \cite{brockwell04}, \cite{brockwell15}, \cite{brockwell2001levy} and \cite{brockwell2014recent}). The simplified version of a CARMA($p,q$) process $Y(t)$ where $\boldsymbol{b} = \boldsymbol{e}_1$, meaning $q=1$ and 
\begin{align}
    \label{eq:strat_temp_model}
    Y(t) = X_1(t),
\end{align}
is called a CAR($p$) process. Note that $p$ is the direct time lag dependence in $Y(t)$. As seen in for example \cite{benth05} and \cite{benth09}, the CAR($p$) model framework suitable for capturing surface temperature and wind evolution. Therefore, it is used in modeling of weather dynamics, for example in relation to financial weather contracts. In the current paper, it will be proved that the deseasonalized stratospheric temperature, $Y(t)$, can be modeled by a CAR($p$) process as in Eq.\,\eqref{eq:strat_temp_model}. 

Now, for a discrete time framework version, consider an AR($p$) process given by
\begin{align}
    X(t) = \beta_1X(t-1) + \beta_2X(t-2) + \cdots + \beta_pX(t-p) + e(t),
    \label{eq:ar(p)-process}
\end{align}
where $X(t)\in\mathbb{R}$ is the value of the AR process at times $t=0,1,\ldots$, and $\beta_k$, $k=1,\ldots , p$, are constant coefficients and $e(t)$ are i.i.d.\ random error terms. The dynamics of a L{\'e}vy-driven CARMA($p,q$) process, see Eq.\,\eqref{eq:car(p)-model}, can be expressed as
\begin{align}
\label{eq:CAR_AR}
	dX_q(t) = 
  \begin{cases} 
  X_{q+1}(t)dt  & \text{if } q = 1,\ldots , p-1 \\
   \displaystyle -\sum_{q=1}^{p} \alpha_{p-q+1}X_{q}(t)dt + \sigma(t-) dL(t)  & \text{if } q=p.
  \end{cases}
\end{align}
By discretization of the expression in Eq.\,\eqref{eq:CAR_AR}, see \cite{benth2008} and \cite{benth13}, a transformation relation between (discrete time) AR($p$) processes and the corresponding (continuous time) CAR($p$) processes is obtained. That is, a transformation relation between $X(t)$ in Eq.\,\eqref{eq:ar(p)-process} and $Y(t)$ in Eq.\,\eqref{eq:strat_temp_model}. See Ch.\,$10$ of \cite{benth2008} for a detailed derivation. Note that the connection between the AR and CAR processes is primarily useful because the continuous version, CAR, allows for deriving analytical results more easily, via stochastic analysis. For instance, \cite{benth2008} uses CAR processes to model surface temperature, which is later used to price options. 
To the best of our knowledge, financial products based directly on stratospheric data are not commonly available. However, as the state of the stratosphere is connected to long-term surface weather forecasting, the CAR model may be of interest for pricing financial weather contracts with long-term maturity. 
Further developments may also aim for a stratospheric temperature model where a control is involved. This means a situation where one may affect the stratospheric temperature directly, or indirectly, via for example carbon emissions. 

Now consider the special case when $p=4$, which will be proven to be well suited for modeling of stratospheric temperature, as assumed from observations in Sect.\,\ref{sec:autoreg_behav}. 
The dynamics of the CAR($4$) process, see Eq.\,\eqref{eq:car(p)-model}, can be written as
\begin{align*}
      \begin{bmatrix}
       dX_1(t)\\
       dX_2(t)\\
       dX_3(t)\\
       dX_4(t)\\
     \end{bmatrix}
      &= \begin{bmatrix}
       0 & 1 & 0 & 0  \\
       0 & 0 & 1 & 0  \\
       0 & 0 & 0 & 1 \\
        -\alpha_4 & -\alpha_3 & -\alpha_2 & -\alpha_1 \\
     \end{bmatrix}
     \cdot
      \begin{bmatrix}
       X_1(t)dt\\
       X_2(t)dt\\
       X_3(t)dt\\
       X_4(t)dt\\
     \end{bmatrix}
     +  \begin{bmatrix}
       0\\
       0\\
       0\\
       \sigma(t-)dL(t)\\
     \end{bmatrix}\\[\bigskipamount] 
     &= \begin{bmatrix}
       X_2(t)dt\\
       X_3(t)dt\\
       X_4(t)dt\\
       -\alpha_4 X_1(t)dt - \alpha_3 X_2(t)dt -\alpha_2 X_3(t)dt-\alpha_1 X_4(t)dt+ \sigma(t-)dL(t)\\
     \end{bmatrix}.
\end{align*}
Note that the dynamics have the form as described in Eq.\,\eqref{eq:CAR_AR}. By the transformation relation between AR($p$) processes and CAR($p$) processes, see \cite{benth2008}, it is found that the model coefficients of the CAR($4$) process are given by
\begin{equation}
\begin{array}{lll}
\alpha_1 &= 4-\beta_1, \alpha_2 = -3\beta_1 - \beta_2 + 6, \\
\alpha_3 &= -3\beta_1 - 2\beta_2 - \beta_3 + 4, \alpha_4 = - \beta_1 - \beta_2 - \beta_3 - \beta_4 + 1.
    \label{eq:ar_car_connection}
\end{array}
\end{equation}

The matrix $A$, see Eq.\,\eqref{eq:ou_matrix_p}, is referred to as the speed of mean reversion throughout the paper. This concept was introduced through half-life computations for Brownian motion-driven (one-dimensional) OU processes in \cite{clewlow00}, Sect.\,$2.4$. That is, for some $s>t$ and a drift coefficient $\alpha$, the formula
\begin{align}
\label{eq: formula}
     s-t = \frac{\ln(2)}{\alpha}
\end{align}
gives the time until a shock $X_{(1)}(t)$ away from the process' long-term mean returns half-way back to this long-term mean (see also \cite{benth13}). For an OU process, the drift coefficient $\alpha$ is the only variable affecting the half-life. As large $\alpha$ gives shorter half-life, and smaller $\alpha$ gives longer half-life, $\alpha$ is referred to as speed of mean reversion. 
In the current paper, non-Gaussian CAR (CARMA) processes are considered rather than standard OU processes. The half-life formula for non-Gaussian CARMA processes is state-dependent, see \cite{benth13}, meaning that the time $s$ in Eq. \eqref{eq: formula} is a stopping time. Denote this stopping time by $\tau$. The special case when the non-Gaussian CARMA process is a CAR process (the process considered in the remaining parts of the paper) gives a half-life formula of the form
\begin{align}
    \label{eq:CAR_speed}
    \boldsymbol{e}_1'\left(\exp(A(\tau-t)) - \frac{1}{2}I\right)\boldsymbol{X}(t) = 0,
\end{align}
see \cite{benth_speed_mean_reversion_note}. Solving this equation for $\tau$ analytically is difficult, and hence it is not clear how the coefficient matrix $A$ affects the process' half-life. The coefficient matrix $A$ will still be referred to as the speed of mean reversion, where each matrix element $\alpha_i$ is assumed to be a contribution to the speed of mean reversion. 

A CAR($4$) model driven by the multidimensional OU process in Eq.\,\eqref{eq:car(p)-model} assumes constant speed of mean reversion. In Sect.\,\ref{sec:mean_reversion}, it will be shown that this assumption is not valid for our dataset. The stratospheric temperature data indicates a seasonal varying pattern in speed of mean reversion from month to month. Based on this observation, an extended model framework is proposed. That is, a CAR($4$) model driven by a multidimensional OU process with time varying speed of mean reversion. The theorem below gives an explicit formula for the (unique) solution of the multidimensional OU SDE driven by a L{\'e}vy process with time-dependent speed of mean reversion.
\begin{theorem}
\label{thm:vectorial_OU_4}
Let $\boldsymbol{X}(t)$ be given by the multidimensional OU process
\begin{align}
\label{eq:CAR_t_dep_speed}
    d\boldsymbol{X}(t)= A(t)\boldsymbol{X}(t)dt + \boldsymbol{e}_4\sigma(t-)dL(t),
\end{align}
where $A(t)$ is the $4\times 4$-matrix
\begin{align}
    \label{eq:ou_matrix_tdep}
    A(t) = \begin{bmatrix}
       0 & 1 & 0 & 0       \\
       0 & 0 & 1 & 0        \\
       0 & 0 & 0 & 1 \\
       -\alpha_4(t) & -\alpha_3(t) & -\alpha_2(t) & -\alpha_1(t) \\
     \end{bmatrix}.
\end{align}
Then 
\begin{align}
    \label{eq:solution_OU_4}
    \boldsymbol{X}(t) =& \exp\left(\int_0^t A(s)ds\right)\boldsymbol{x} + \exp\left(\int_0^t A(s)ds\right)\int_0^t \exp\left(-\int_0^s A(u)du\right)\boldsymbol{e}_4\sigma(s-)dL(s),
\end{align}
where $\boldsymbol{x}\triangleq \boldsymbol{X}(0)$.
\end{theorem}
\begin{proof}
Rewrite the SDE in Eq.\,\eqref{eq:CAR_t_dep_speed} by use of the It{\^o}-L{\'e}vy decomposition, to find that 
\begin{align*}
    d\boldsymbol{X}(t) =& A(t)\boldsymbol{X}(t)dt + \boldsymbol{e}_4\sigma(t-)\left(  a dt + bdB(t) + \int_{\mathbb{R}}z\boldsymbol{\bar{N}}(dt,dz) \right)\\
    =& \left(A(t)\boldsymbol{X}(t) + \boldsymbol{e}_4a\sigma(t-) \right)dt + \boldsymbol{e}_4b\sigma(t-)dB(t) + \boldsymbol{e}_4\sigma(t-)\int_{\mathbb{R}}z\boldsymbol{\bar{N}}(dt,dz).
\end{align*}
By definition, $\boldsymbol{X}(t)$ is a multidimensional It\^o-L{\'e}vy process. Apply the multidimensional It\^o formula on $\displaystyle d\boldsymbol{Y}(t)\triangleq d\left( \exp\left(-\int_0^t A(s)ds\right) \boldsymbol{X}(t) \right)$. By defining $Y(t)\triangleq f(t,\boldsymbol{X}(t))$, it is found by the dominated convergence theorem and the fundamental theorem of calculus that 
\begin{align*}
    \frac{\partial f}{\partial t}(t,\boldsymbol{X}(t)) =& \frac{\partial}{\partial t}\sum_{k=1}^{\infty}\frac{1}{k!}\left(-\int_0^t A(s) ds\right)^k\boldsymbol{X}(t)\\
    =& -\sum_{k=1}^{\infty}\frac{1}{(k-1)!} \left(-\int_0^t A(s) ds\right)^{k-1}A(t)\boldsymbol{X}(t)\\
    =& -\exp\left(-\int_0^t A(s)ds\right)A(t)\boldsymbol{X}(t).
\end{align*}
Furthermore, note that 
\begin{align*}
    \frac{\partial f}{\partial x_i}(t,\boldsymbol{X}(t)) = \exp\left(-\int_0^t A(s)ds\right)\boldsymbol{e}_i\quad\text{and}\quad \frac{\partial^2 f}{\partial x_i\partial x_j}(t,\boldsymbol{X}(t)) = 0
\end{align*}
for all $i$ and $i,j$ respectively. The remaining terms coming from the It{\^o} formula are trivial. Thus, one finds that
\begin{align*}
    d\boldsymbol{Y}(t) =& -\exp\left(-\int_0^t A(s)ds\right)A(t)\boldsymbol{X}(t) dt +\exp\left(-\int_0^t A(s)ds\right)\left( X_2(t)\boldsymbol{e}_1 + X_3(t)\boldsymbol{e}_2 + X_4(t)\boldsymbol{e}_3 \right.\\
    &\quad\quad\left. +(-\alpha_4(t)X_1(t) - \alpha_3(t)X_2(t) - \alpha_2(t)X_3(t) - \alpha_1(t)X_4(t) + a\sigma(t-))\boldsymbol{e}_4\right)dt\\
    &+\exp\left(-\int_0^t A(s)ds\right)\boldsymbol{e}_4b\sigma(t-)dB(t)\\
    &+ \exp\left(-\int_0^t A(s)ds\right)\int_{\abs{z}<R}\left\{ \boldsymbol{X}(t-)+\boldsymbol{e}_4\sigma(t-) z - \boldsymbol{X}(t-) - \boldsymbol{e}_4\sigma(t-) z\right\}\nu(dz)dt\\
    &+ \exp\left(-\int_0^t A(s)ds\right)\int_{\mathbb{R}}\left\{ \boldsymbol{X}(t-) + \boldsymbol{e}_4\sigma(t-)z - \boldsymbol{X}(t-)  \right\}\boldsymbol{\bar{N}}(dt,dz)\\
    =& \exp\left(-\int_0^t A(s)ds\right)\left(A(t)\boldsymbol{X}(t) - A(t)\boldsymbol{X}(t)+ \boldsymbol{e}_4\sigma(t-)\left( adt + bdB(t) + \int_{\mathbb{R}}z\boldsymbol{\bar{N}}(dt,dz)  \right)\right)\\
    =& \exp\left(-\int_0^t A(s)ds\right)\boldsymbol{e}_4\sigma(t-)dL(t).\\\end{align*}
 Hence, from the definition of $Y(t)$, when $\boldsymbol{x}\triangleq \boldsymbol{X}(0)$,
 \begin{align*}
\boldsymbol{X}(t) =& \exp\left(\int_0^t A(s)ds\right)\boldsymbol{x} + \exp\left(\int_0^t A(s)ds\right)\int_0^t \exp\left(-\int_0^s A(u)du\right)\boldsymbol{e}_4\sigma(s-)dL(s).
\end{align*}
\end{proof}


\section{\label{sect:main_analysis_ch}Stochastic modeling of daily-zonal mean stratospheric temperature}
The aim of the following sections is to fit a CAR model to the daily-zonal mean stratospheric temperature data obtained from the ECMWF ERA-Interim reanalysis product. 


\subsection{Methodology for deriving and fitting a stochastic model to stratospheric temperature data\label{sect:methodology}}
This section describes the data analysis applied in Sects.\,\ref{sec:season_function_data}-\ref{sect:residuals_analysis} to fit the model in Eq.\,\eqref{eq:general_temp_model} to ERA-Interim stratospheric temperature reanalysis data (described in \cite{ecmwf_stratospheric_temp} and specified in Sect.\,\ref{sec:strat_temp_data_retrieval}). Applying this methodology shows that the model in Eq.\,\eqref{eq:general_temp_model} is suitable to model stratospheric temperature when $Y(t)$ is a non-Gaussian CAR($4$) process. 

Assume that a dataset of stratospheric temperatures indexed by time is given, and denote this by $\mathcal{S}$. A detailed description of the stratospheric temperature dataset used in this paper will be given in Sect.\,\ref{sec:strat_temp_data_retrieval}. The main steps of the data analysis of $\mathcal{S}$ are:
\begin{enumerate}
    \item Fit a deterministic continuous seasonality function $\Lambda(t)$ to $\mathcal{S}$. Subtract $\Lambda(t)$ from $\mathcal{S}$ to obtain a dataset of deseasonalized stratospheric temperatures, denoted $\mathcal{S}_d$.
    \item Fit an AR($p$) model to $\mathcal{S}_d$ with the choice of $p$ based on the PACF of the dataset. Subtract the fitted AR($p$) model from $\mathcal{S}_d$ to obtain a dataset of residuals, $\mathcal{E}$.
    \item Compute the empirical expected values of the squared residuals each day over the year (assumed to be $365$ days) to construct an approximation of the time-varying volatility function, $\sigma(t)$.
    \item Divide $\mathcal{E}$ by $\sigma(t)$ to obtain a dataset of $\sigma(t)$-scaled residuals, denoted $\hat{\mathcal{E}}$. %
    Find the probability distribution of the elements in $\hat{\mathcal{E}}$ (by statistical analysis).
\end{enumerate}
As the goal of this work is to obtain a continuous time stochastic model for stratospheric temperature, notation corresponding to continuous functions will be used in the more detailed explanation that follows.

Assume that the stratospheric temperature, $S(t)$, is given by Eq.\,\eqref{eq:general_temp_model}, $Y(t)$ being a CAR($p$) process as in Eq.\,\eqref{eq:strat_temp_model}. The lag $p=4$ is chosen based on observations in Sect.\,\ref{sec:autoreg_behav}, meaning that the direct memory effect reaches over four days. Then, by the transformation relation between CAR($4$) and AR($4$) processes in Eq.\,\eqref{eq:ar_car_connection}, the deseasonalized temperature is given by
\begin{align}
\label{eq:ds_strat_temp_formula}
X(t) =	S(t) - \Lambda(t) = \sum_{k=1}^4 \beta_k X(t-k) + e(t).
\end{align}
The seasonality function $\Lambda(t)$ is fit by least squares to simulate the seasonal behavior of the stratospheric temperature data in $\mathcal{S}$, and then subtracted from the stratospheric temperature data to find the discrete version of deseasonalized stratospheric temperature. %
The deseasonalized temperature is given by the AR($4$) process $X(t)$ in Eq.\,\eqref{eq:ds_strat_temp_formula}, with random error terms (residuals) $e(t)$. Therefore, by use of least squares, an AR($4$) model is fit to the deseasonalized stratospheric temperature data in $\mathcal{S}_d$, and then subtracted to find the residuals dataset $\mathcal{E}$. Mathematically, the residuals are given by
\begin{align}
\label{eq:residuals_formula}
e(t)= X(t)-\sum_{k=1}^4 \beta_k X(t-k).
\end{align} 
In Sect.\,\ref{sect:residuals_analysis}, yearly heteroskedasticity is observed in the squared residuals. This means that the daily variance values (over the year) of the dataset $\mathcal{E}$ are time-dependent. Therefore, a time-varying volatility function $\sigma(t)$ is approximated and divided on the residuals in $\mathcal{E}$ to obtain the $\sigma(t)$-scaled residuals
\begin{align}
\label{eq:deseasonalized_residuals}
\epsilon(t) = \frac{e(t)}{\sigma(t)}.
\end{align}
That is, $\epsilon(t)$ represents the data in $\hat{\mathcal{E}}$. Recall from Sect.\,\ref{sect:car_ar_connection} (Eq.\,\eqref{eq:ar(p)-process}), that $\epsilon(t)$ are i.i.d.\ random variables. The mean value of residuals each day $d$ during the year, $d\in[1,365]$ (see Sect.\,\ref{sec:strat_temp_data_retrieval}), is assumed to be constant. Therefore, the variance each day during the year is given by
\begin{align}
\label{eq:daily_variance}
\text{Var}(e_d(t)) = \left(E[e_d^2(t)] - E[e_d(t)]^2\right),
\end{align}
where $E[e_d(t)]$ represents the empirical mean value of residuals at day $d$. The magnitude of $E[e_d(t)]^2$ is insignificant compared to $E[e_d^2(t)]$, see Fig.\,\ref{fig:compare_var} for an illustration of this. Hence, the approximation 
\begin{align}
    \label{eq:daily_variance_estimate}
    \text{Var}(e_d(t))\simeq E[e_d^2(t)]
\end{align}
is used to fit an appropriate time-varying variance function, $V(t)$, for $t\in[1,365]$. See Sect.\,\ref{sect:residuals_analysis} for a thorough explanation of how to compute $E[e_d^2(t)]$ empirically.
When $E[e_d^2(t)]$ is computed for each $d$, the function $V(t)$ is fit to the values by use of three heavily truncated Fourier series (by least squares), which are connected by two sigmoid functions. The time-varying volatility function is finally computed as $\sigma(t) = \sqrt{V(t)}$.
When $\hat{\mathcal{E}}$ is obtained, an appropriate probability density function (pdf) describing the distribution of the $\sigma(t)$-scaled residuals has to be found. Finally, the last step is to introduce a stochastic process which is able to replicate the behaviour of the particular pdf. The $\sigma(t)$-scaled residuals function $\epsilon(t)$ is represented by this stochastic process in the CAR model. 

From this analysis, an appropriate driving stochastic process for the CAR($4$) model is obtained. However, at this point (due to doing time series analysis) the model is given by Eq.\,\eqref{eq:ds_strat_temp_formula}, and is thus a discrete model. The CAR($4$) model, which is given in Eq.\,\eqref{eq:strat_temp_model}, is found by applying the transformation relation in Eq.\,\eqref{eq:ar_car_connection}. A detailed description of this data analysis methodology applied to the ERA-Interim stratospheric temperature reanalysis data (see Sect.\,\ref{sec:strat_temp_data_retrieval}) is given in Sects.\,\ref{sec:season_function_data}-\ref{sect:residuals_analysis}.
\begin{figure}[ht!]
	\centering
    \includegraphics[scale=0.6]{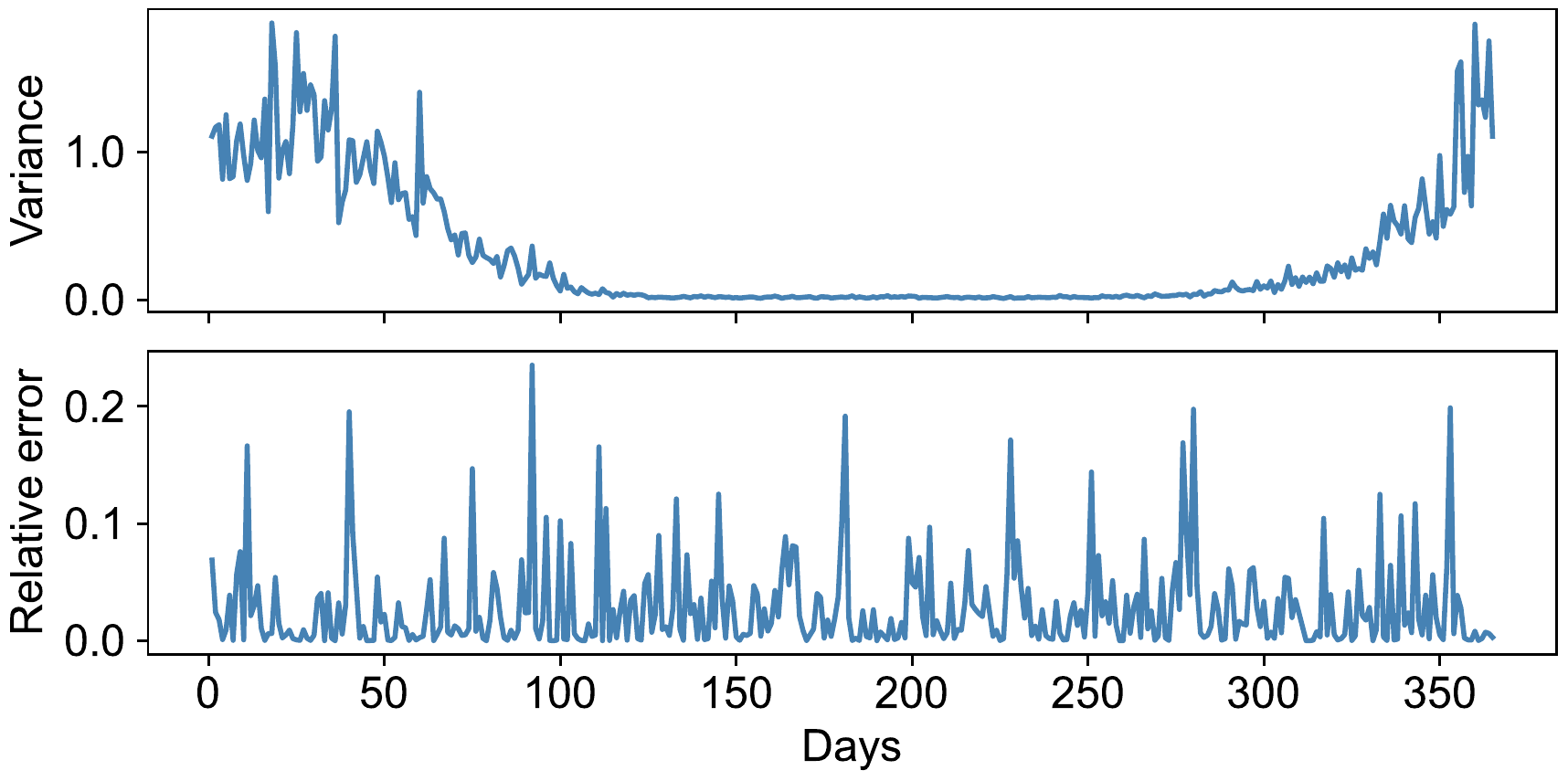}
    \caption{Variance each day of the year computed with the approximation in Eq.\,\eqref{eq:daily_variance_estimate}, and the corresponding relative error when compared with the definition of variance (Eq.\,\eqref{eq:daily_variance}). The mean absolute percentage error is $2.8$\%} 
    \label{fig:compare_var}
\end{figure}


\subsection{\label{sec:strat_temp_data_retrieval}The zonal mean stratospheric temperature dataset}
The aim of this section is to describe the stratospheric temperature dataset analyzed in the remaining of this paper.

Define a spherical coordinate system such that $(r_{\text{tot}},\theta,\varphi)$ represents a point in the atmosphere. Let $r_{\text{tot}}=r_{E}+r$ represent the altitude from the center of the Earth, where $r_E$ is the radius of the Earth and $r$ is the distance from Earth's surface to the atmospheric point of interest. Further, $\theta$ represents the longitude and $\varphi$ the latitude. With the presented notation, the region of interest in this paper can be defined as 
\begin{align*}
    R \triangleq\big\{(r_{\text{tot}},\theta,\varphi) | r=10\;\text{hPa} \simeq 30\;\text{km}, \theta\in[-180^{\circ}\text{E},180^{\circ}\text{E}),\varphi=60^{\circ}\text{N}\big\}.
\end{align*}
The region $R$ is an area bounded by a circumpolar line in the extra-tropical stratosphere. The pressure level $10$\;hPa corresponds roughly to $30$\;km altitude. Stratospheric dynamics in this region are highly variable, and they depend on the state of the stratospheric polar vortex. 

Enhanced probing and representation of the stratosphere in atmospheric models and numerical weather prediction systems has potential to enhance surface weather predictions on weekly to monthly timescales (e.g.,  \cite{Pedatella2018}, \cite{KarpechkoWMO2016} and \cite{butler2019chapter}). Maybe the most striking example of stratospheric influence on the surface is the extreme event of sudden stratospheric warmings (SSWs), where an abrupt disruption in the stratospheric winter circulation occurs, accompanied by a stratospheric temperature increase of several tens of degrees. SSWs are detectable in the region $R$. Through stratosphere-troposphere coupling, the effects of SSWs can extend to the troposphere, with increased probability of shifts in the jet stream and storm tracks, further affecting the expected precipitation and surface temperatures. This phenomenon can, for example, be manifested as harsher winter weather regimes on continental North America and Eurasia, see \cite{baldwin2021sudden}. This may have impact on several sectors in society and industry.

With the purpose of deriving a stochastic stratospheric temperature model, a dataset $\mathcal{D}$ from the  ERA-Interim reanalysis model product is retrieved from ECMWF (see \cite{ecmwf_stratospheric_temp}), such that each temporal data point $d_i\in\mathcal{D}$ represents the spatial mean stratospheric temperature over $R$. As the spatial mean is taken over the full circumpolar interval, this is denoted as the zonal mean. The zonal mean properties of the stratosphere at the $10$\;hPa pressure level is commonly considered in stratospheric diagnostics, and when studying stratospheric events like SSWs and beyond \citep{butler2015defining}. The subscript $i$ represents a measurement every six hours from midnight, and the zonal mean is taken over $\theta$ at a $0.5^{\circ}$ spacing. 
That is, $\mathcal{D}$ contains four zonal mean temperature measurements every day within the interval $T\in[1\text{ January }1979,31\text{ December }2018]$. For computational convenience, all data from 29 February each leap year are excluded from $\mathcal{D}$, such that the length of each year is constant. %
All stated specifications of $\mathcal{D}$ are collected in Tab.\,\ref{tab:strat_temp_specifications}. Further, define the dataset 
\begin{align*}
    \mathcal{S} \triangleq\big\{S_k: S_k= E[d  | D_k]\big\},
\end{align*}
where $D_k$ are subsets of $\mathcal{D}$ containing four data points every given day $k$ in the time interval $T$ (except days 29 February), and $E[d | D_k]$ is the empirical mean. 
That is, $\mathcal{S}$ contains daily-zonal mean stratospheric temperatures over the region $R$ for days in the time interval $T$. This is the time series analyzed in the remaining of this paper. As 29 February is excluded each leap year, the total number of data points in $\mathcal{S}$ is $14,600$. A plot of $\mathcal{S}$ for the last ten years (from 1 January 2009 to 31 December 2018) with a fitted seasonality function was shown in Fig.\,\ref{fig:temp_10y_with_seasonfit}.
\begin{table}[ht!]
    \centering
    \caption{Specifications of the stratospheric temperature dataset $\mathcal{D}$ \label{tab:strat_temp_specifications}}
  \begin{tabular}{ c c c c c c }
  \toprule
  Date & Grid & Pressure level & Time & Area & Unit\\
  \midrule
\makecell{1 January 1979 to\\31 December 2018} \hspace{1mm} & \hspace{1mm} $0.5^{\circ}$ \hspace{1mm} & \hspace{1mm} $10$\;hPa \hspace{1mm} & \hspace{1mm} \makecell{00:00, 06:00,\\12:00, 18:00} \hspace{1mm} & \hspace{1mm} \makecell{$60^{\circ}$N and\\$[-180^{\circ}\text{E},180^{\circ}\text{E})$}  \hspace{1mm}& \hspace{1mm} Kelvin\\
  \bottomrule
\end{tabular}
\end{table}


\subsection{\label{sec:season_function_data}Fitting a seasonality function to stratospheric temperature data}
In this section, seasonality in stratospheric temperature data is analyzed. Seasonality in this setting means (deterministic) periodically repetitive patterns of temperature dynamics over time. A deterministic seasonality function will be fit to the dataset $\mathcal{S}$, with the aim of further analyzing deseasonalized temperature data where these periodically repetitive patterns are removed. 
Although the stratosphere is typically characterized by variations on longer timescales than the troposphere, there are oscillation, or atmospheric tide, patterns present at these altitudes as well. In addition to the directly forced cycles causing seasonal effects, for example, the phenomenon of quasi-biennial oscillations (QBO) is a nearly periodic phenomenon in the stratosphere. The QBO period is variable, but averages to about $28$ months. This is a phenomenon occurring in the equatorial stratosphere. Still, the QBO can affect stratospheric conditions from pole to pole, and even has effects on the breaking of wintertime polar vortices, leading to SSWs, see \cite{vallis2017} and \cite{baldwin2001}. 

Seasonal effects complicate stochastic modeling because they cause non-stationarity. In the current study, daily-zonal mean stratospheric temperatures over $40$ years are considered, meaning that the periodic phenomena of interest are the yearly cycle and the QBO. Non-stationarity can also result from long-term effects of greenhouse gases and ozone, anthropogenic forcings that cause a stratospheric cooling trend, see \cite{cnossen2015}, \cite{danilov20} and \cite{steiner20}. The first step in deriving a stochastic stratospheric temperature model is to fit a seasonality function to the data, and then to subtract this to remove the non-stationary effects. 

As the yearly cycle is the most pronounced phenomenon (Figs.\,\ref{fig:temp_10y_with_seasonfit} and \ref{fig:acf_strat_temp}), a Fourier series with a period of $365$ days is chosen as seasonality function. Further, the long-term decreasing trend in stratospheric temperature is approximately linear,
meaning that a linear function should be present in the seasonality function as well. Based on these considerations, a seasonality function $\Lambda(t)$ is defined as in Eq.\,\eqref{eq:lambda_def}. 

In the following, the continuous version of daily-zonal mean stratospheric temperature data $S_i\in\mathcal{S}$ (as described in Sect.\,\ref{sec:strat_temp_data_retrieval}) is denoted by $S(t)$, where $t\in\mathbb{R}^{+}$. Let $S(t)$ be given by the stochastic model in Eq.\,\eqref{eq:general_temp_model} with seasonality function 
\begin{align}
    \label{eq:lambda_def}
    \Lambda(t) = c_0 + c_1t + \sum_{k=1}^{n} \left( c_{2k}\cos(k\pi t/365) + c_{2k + 1}\sin(k \pi t/365) \right),
\end{align}
where $c_0,c_1,c_{2},\ldots ,c_{2n+1}$ are constants. The choice of $\Lambda(t)$ is made based on the discussion above, where $c_1$ captures the slope of the long-term cooling of the stratosphere (corresponding to global warming of the troposphere), and where the constant term $c_0$ represents the average level at the beginning of the time series $\mathcal{S}$. 
The constants $c_2,\ldots ,c_{2k+1}$ describe the yearly cycle as weights in the truncated Fourier series. 

The seasonality function $\Lambda (t)$ is fit to the time series in $\mathcal{S}$ with $n=10$ (using least squares), with resulting parameters given in Tab.\,\ref{tab:seasonality_function_parameters}. Fig.\,\ref{fig:temp_10y_with_seasonfit} displays the fitted seasonality function $\Lambda(t)$ together with the ten last years of the times series in $\mathcal{S}$. The value of $c_0$ in Tab.\,\ref{tab:seasonality_function_parameters} indicates that the daily-zonal mean stratospheric temperature over the region $R$ (see Sect.\,\ref{sec:strat_temp_data_retrieval}) was approximately $226.15$\;K ($-47.00^\circ$C) in 1979. The negative value of $c_1$ confirms the long-term cooling effect of the stratosphere. The $c_1$ value found corresponds to a daily-zonal mean stratospheric temperature decrease of approximately $1.05$\;K (equivalent to a change of $1.05^\circ$C) over the last $40$ years at $60^\circ$N and $10$\;hPa. This is consistent with \cite{steiner20}, estimating an overall cooling of the stratosphere of about $1$-$3$\;K over the same time span.
\begin{table}[ht!]
    \centering 
    \caption{Seasonality function, $\Lambda(t)$, parameters (see Eq.\,\eqref{eq:lambda_def}) for daily-zonal mean stratospheric temperature at $60^\circ$N %
    and $10$\;hPa between 1 January 1979 and 31 December 2018 
    \label{tab:seasonality_function_parameters}}
    \begin{tabular}{c c c c c c c c c c c}
  \toprule
	$c_0$ &   $c_2$  & $c_4$ & $c_6$ & $c_8$ & $c_{10}$ &   $c_{12}$ & $c_{14}$ & $c_{16}$ & $c_{18}$ & $c_{20}$\\
	\midrule
	$226.15$ & $-0.05$ & $-12.09$ & $0.23$ & $1.88$ & $0.33$ & $0.16$ & $0.13$ & $-0.09$ & $-0.15$ & $-0.01$\\
  \bottomrule
  \toprule
  	$c_1$ &   $c_3$  & $c_5$ & $c_7$ & $c_9$ & $c_{11}$ &   $c_{13}$ & $c_{15}$ & $c_{17}$ & $c_{19}$ & $c_{21}$ \\
  	\midrule
	$-0.000072$ & $-0.11$ & $1.63$ & $-0.23$ & $2.81$ & $-0.04$ & $1.54$ & $0.14$ & $0.45$ & $0.05$ & $0.11$\\
  \bottomrule
\end{tabular}
\end{table}


\subsection{\label{sect:AR-model}Fitting an AR model to deseasonalized stratospheric temperature data}
Having deseasonalized the stratospheric temperature dataset $\mathcal{S}$, the next step is to fit an AR model to the deseasonalized dataset $\mathcal{S}_d$.
Temperature tends to having a mean-reverting property over time, a property that can be modeled by an AR($p$) process \cite{benth2008}. Based on the discussion in Sect. \ref{sec:autoreg_behav}, suppose that the deseasonalized stratospheric temperature $Y(t)=S(t) - \Lambda(t)$ can be modeled by an AR($p$) process as represented in Eq.\,\eqref{eq:ar(p)-process}, where the random error terms $e(t)$ represent the model residuals. The empirical ACF and PACF of the deseasonalized stratospheric temperature data illustrated in Figs.\,\ref{fig:acf_strat_temp_ds} and \ref{fig:pacf_strat_temp_ds} confirm that it is appropriate to model $Y(t)$ by an AR($p$) process, and indicate that $p=4$ is needed to explain the time series evolution, see \cite{Levendis18}. 
An AR($4$) model is fit to the deseasonalized stratospheric temperature data in $\mathcal{S}_d$ by use of least squares. The resulting AR($4$) model parameters are presented in Tab.\,\ref{tab:ar/car4_parameters}. By use of Eq.\,\eqref{eq:ar_car_connection} the corresponding CAR($4$) process is calculated, and the resulting model parameters for this continuous model are presented in Tab.\,\ref{tab:ar/car4_parameters} as well. Based on the reasoning in \cite{benth2008}, preservation of stationarity of the CAR($4$) model depends on the properties of the time-dependent volatility function $\sigma(t)$.
However, as long as all eigenvalues of the matrix $A$ have negative real part, it is ensured that the modeled temperature on average will coincide with the seasonality function $\Lambda(t)$ when time approaches infinity. This is because, as will be shown in Sect.\,\ref{sect:residuals_analysis}, the L{\'e}vy-driven CAR($4$) model generates NIG distributed random variables with mean zero, a property which is preserved for the model in the long run when the eigenvalues have negative real part. The eigenvalue equation 
\begin{align}
\label{eq:eigenvalueeq_ar(4)}
    \lambda^4 + \alpha_1\lambda^3 + \alpha_2\lambda^2 + \alpha_3\lambda +  \alpha_4
\end{align}
has roots $\lambda_{1,2} = -1.01\pm 0.43i$, $\lambda_{3} = -0.36$ and $\lambda_4 = -0.07$, and the stationarity condition is therefore satisfied. 
\begin{table}[ht!]
    \centering
    \caption{AR($4$) model parameters, and parameters of its continuous counterpart, when fitted to daily-zonal mean stratospheric temperature over $60^{\circ}$N and $10$\;hPa in the period 1 January 1979 to 31 December 2018
    \label{tab:ar/car4_parameters}}
    \begin{tabular}{ c c c c c }
  \toprule
  \multirow{-0.5}{*}{AR($4$) parameters} &
    $\beta_1$ & $\beta_2$ & $\beta_3$ & $\beta_4$  \\ \cline{2-5}
  & $1.55$ & $-0.75$ & $0.28$ & $-0.11$\\
  \bottomrule
  \toprule
  \multirow{-0.5}{*}{CAR($4$) parameters} & $\alpha_1$ & $\alpha_2$ & $\alpha_3$ & $\alpha_4$ \\ \cline{2-5}
  & $2.45$ & $2.10$ & $0.56$ & $0.03$\\
  \bottomrule
\end{tabular}
\end{table}


\subsection{\label{sect:residuals_analysis}Analyzing the residuals}
In this section, the residuals (random error terms) in the dataset $\mathcal{E}$ are analyzed to determine the appropriate stochastic driving process of the CAR($4$) model for stratospheric temperature.

In computing the parameters of the CAR($4$) model in Sect.\,\ref{sect:AR-model}, the deterministic mean-reverting property of the stratospheric temperature is found. A suitable stochastic driving process for the model residuals, corresponding to the random error terms in Eq.\,\eqref{eq:ar(p)-process}, still remains to be found. As derived in Sect.\,\ref{sect:methodology}, the model residuals are given by 
\begin{align}
e(t) = X(t) - \sum_{k=1}^{4}\beta_k X(t-k).
\label{eq:strat_temp_model_residuals}
\end{align}
The approach to find a suitable stochastic driving process is therefore to empirically determine the stratospheric temperature model residual distribution in $\mathcal{E}$. In \cite{Levendis18}, there is a statement that residuals of AR($p$) models approach white noise for larger $p$. Therefore, it is reasonable as a first guess to assume that $e(t)$ is distributed as i.i.d.\ $N(0,1)$. A normal fit is performed on $\mathcal{E}$, however, by Fig.\,\ref{fig:distr_res} it is clear that the data is not normally distributed. Further, Fig.\,\ref{fig:acf_res} indicate a seasonally varying empirical ACF of squared residuals. Since the distributional mean value is close to zero, this is a sign of seasonal heteroskedasticity in the distributional variance (see Sect.\,\ref{sect:methodology}). A similar seasonal pattern in the empirical ACF of squared residuals was observed by \cite{benth2008} in daily average surface temperature data in Sweden.

Still assuming the stratospheric temperature model residuals to be normally distributed random variables, however with a time-varying variance rather than the constant $1$, $e(t)$ is rewritten as 
\begin{align}
\label{eq:res_with_seasonality_func}
    e(t) = \sigma(t)\epsilon(t).
\end{align}
Here, $\epsilon(t)$ are distributed as i.i.d.\ $N(0,1)$, and $\sigma(t)$ is a yearly (see Fig.\,\ref{fig:acf_res}) time-varying deterministic function. To adjust for the heteroskedasticity, the volatility function $\sigma(t)$ must be defined explicitly. %
To do so, the same approach as in \cite{benth2008} is used: Daily residuals over 40 years are organized into $365$ groups, one group for each day of the year. This means that all observations on 1 January are collected into group $1$, all observations on 2 January into group $2$, and so on until all days of all years are grouped together. 
Recall that observations on 29 February were removed each leap year, such that each year contains 365 data points. By computing the empirical mean of the squared residuals in each group, an estimate of the expected squared residual each day of the year is found, corresponding to an estimate of the daily variance, as explained in Sect.\,\ref{sect:methodology} (Eq.\,\eqref{eq:daily_variance_estimate}). %
The resulting 365 estimates of daily variance yields an estimate of the time-varying variance function, $V(t)$, over the year. This is illustrated in Fig.\,\ref{fig:expected_sq_res}. The yearly heteroskedasticity is clearly visible. 
\begin{figure}[ht!]
	\centering
	\begin{subfigure}[ht!]{0.49\textwidth}
	\centering
	\includegraphics[scale=0.6]{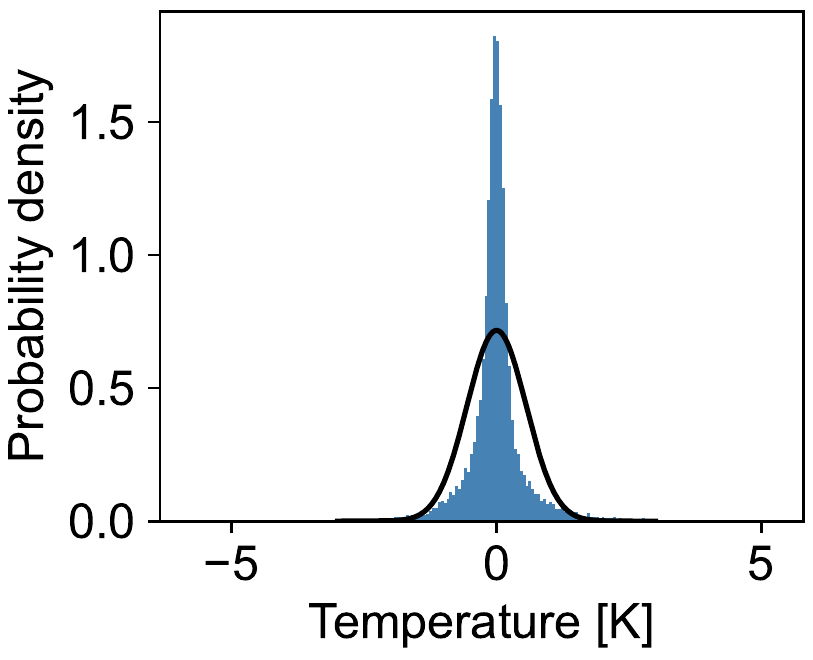}
	\caption{Distribution of $e(t)$}
	\label{fig:distr_res}
	\end{subfigure}
	\begin{subfigure}[ht!]{0.49\textwidth}
	\centering
	\includegraphics[scale=0.6]{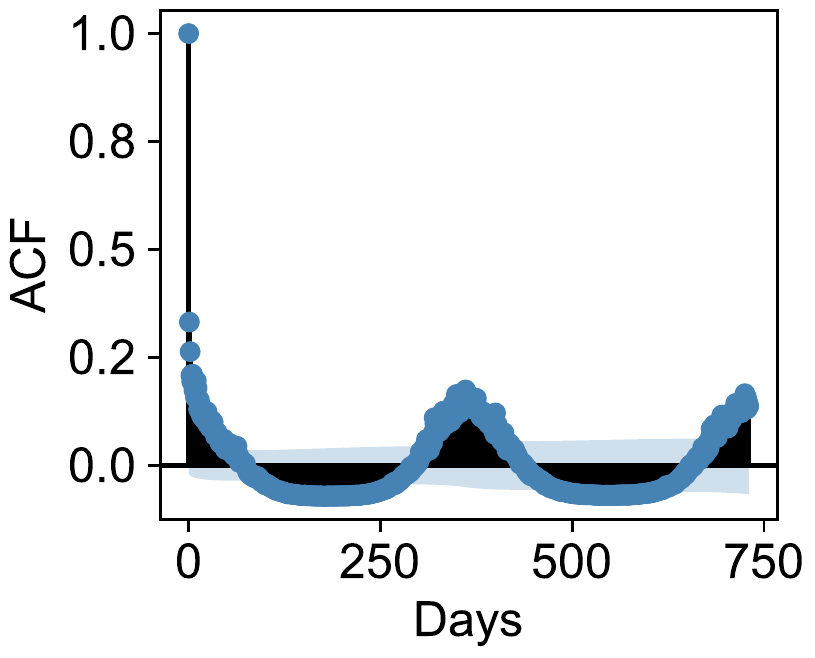}
	\caption{ACF of $e^2(t)$}
	\label{fig:acf_res}
	\end{subfigure}
    \caption{The distribution of stratospheric temperature model residuals $e(t)$ (Eq.\,\eqref{eq:strat_temp_model_residuals}) with a fitted normal distribution, and ACF of the squared model residuals $e^2(t)$}
    \label{fig:distr_res_acf_res}
\end{figure}
Recall that, by definition, the volatility function $\sigma(t)$ is the square root of the time-varying variance function. With the aim of obtaining $\sigma(t)$, an analytic function is fit to the empirically computed expected value of squared residuals, to find a proper function $V(t)$. Fig.\,\ref{fig:expected_sq_res} illustrate that the volatility in the stratospheric temperature variance is much higher in winter time than in summer time, as seen in \cite{haynes2005}. This, as well as the shape of the estimated daily variance (expected squared residuals) over the year, makes function estimation with Fourier series more challenging than simply fitting a single Fourier series. To properly fit a function $V(t)$, the year is split into three parts. Each part represents the winter/spring season, summer season and autumn/winter season, respectively. A local variance test is performed to find appropriate seasonal endpoints. That is, the summer season variance is low and stable compared to the two other seasons, and so the summer season endpoints are set where the local variance hits a given limit, $\delta$. Define the local variance as 
\begin{align*}
v(n) \triangleq \dfrac{1}{2n-1}\displaystyle\sum_{i=0}^{2n-1}(x_i - \mu)^2,
\end{align*}
\noindent where $n$ represents the degree of locality. The number of elements included in the sum is odd, $2n-1$, such that the local variance for each element is based on a symmetric number of neighbours on each side. The test is performed as follows: First, the limit $\delta$ is defined such that mid-summer local variances do not exceed $\delta$. Second, with estimated expected squared residuals given as $[V_1,V_2,\ldots , V_{364},V_{365}]$, the local variance $v(n)_k$ is computed for each point $V_k$ in $[V_{n},V_{n+1},\ldots , V_{364-n},V_{365-n}]$ (each endpoint is cut with $n-1$ elements for computability). Third, an array $K=[k_i \in \{n,n+1,\ldots ,365-n\}:   \mbox{ }v(n)_{k_i}<\delta]$ is constructed (sequentially in time), such that the index (that is, day) of elements with satisfactory small local variance is known. Fourth, based on the array $K$, a collection $\mathcal{K} = \{(k_i-k_{i-1},k_i)\in\mathbb{N}\times\{n,n+1,\ldots , 365-n\}: \mbox{ }\forall \mbox{ } k_i \in K\}$ is constructed for stability purposes. Finally, all pairs in $\mathcal{K}$ where $k_i-k_{i-1}>1$ (day) are printed such that stability of the condition $v(n)_k<\delta$ can be evaluated manually.

The analysis is performed with $n=5$ and $\delta=0.0002$, and gives cutoff at days $115$ and $288$, corresponding to 25 April and 15 October, respectively. For simplicity, the cutoffs are set at the following whole month. That is, the three seasons winter/spring, summer and autumn/winter are defined to be in the intervals $[\text{1 January},\text{30 April}]$, $[\text{1 May},\text{31 October}]$ and $[\text{1 November},\text{31 December}]$, respectively. A function is fit to the estimate of $V(t)$ for each of the three seasons by use of the truncated Fourier series 
\begin{align}
    \label{eq:wf_2sum}
    w_f(t) = d_0 + \sum_{k=1}^{2} \left( d_{2k-1}\cos(fk\pi t/365) + d_{2k}\sin(fk \pi t/365) \right),
\end{align}
where $d_0,...,d_4$ are constants and $f$ is a given parameter adjusting the series frequency. By manual inspection, the function $w_f(t)$ for each of the three seasons are chosen as $w_{0.44(1)}(t)$, $w_{2.0}(t)$ and $w_{0.44(2)}(t)$ respectively, and Tab.\,\ref{tab:res_sq_fitted_parameters} displays the fitted parameters.

The transitions from winter/spring to summer and from summer to autumn/winter should be smooth in order to obtain a smooth yearly time-varying volatility function $\sigma(t):[1,365]\to\mathbb{R}$. This is achieved by connecting the three functions $w_{0.44(1)}(t)$, $w_{2.0}(t)$ and $w_{0.44(2)}(t)$ with two sigmoid functions as in \cite{sigmoid}. The sigmoid function $\omega(x)$ and the connective function $\xi(x)$ are given by 
\begin{align*}
\omega(x) = \frac{1}{1+\exp\left(-(\frac{x-a}{b})\right)} \quad \text{and}\quad  \xi(x) = \big(1-\omega(x)\big)f_1(x) + \omega(x) f_2(x),
\end{align*}
where $a$ and $b$ are shift and scaling constants respectively, and $f_1(x)$ and $f_2(x)$ are two functions that are to be connected. By connecting the functions $w_{0.44(1)}(t)$ and $w_{2.0}(t)$ with $a=120$ and $b=2$, and connecting the functions $w_{2.0}(t)$ and $w_{0.44(2)}(t)$ with $a=304$ and $b=5$, a smooth function $\sigma(t)$ is found. The resulting volatility function is illustrated as $\sigma^2(t)$ (or $V(t)$) in Fig.\,\ref{fig:expected_sq_res}, together with the estimated daily variances during the year.
\begin{table}[ht!]
    \centering
    \caption{Parameters of fitted Fourier series, $w_f(t)$ (see Eq.\;\eqref{eq:wf_2sum}), to each of the three seasons winter/spring, summer and autumn/winter
    \label{tab:res_sq_fitted_parameters}}
    \begin{tabular}{ c c c c c c }
  \toprule
  \multirow{-0.5}{*}{Winter/spring ($f=0.44$)}
    & $d_0$ & $d_1$ & $d_2$ & $d_3$ & $d_4$ \\ \cline{2-6}
  & $-507.58$ & $633.29$ & $269.90$ & $-124.82$ & $-130.75$\\
  \bottomrule
  \toprule
  \multirow{-0.5}{*}{Summer ($f=2.0$)} & $d_0$ & $d_1$ & $d_2$ & $d_3$ & $d_4$ \\ \cline{2-6}
  & $0.092$ & $0.107$ & $0.023$ & $0.034$ & $0.015$\\
  \bottomrule
  \toprule
  \multirow{-0.5}{*}{Autumn/winter ($f=0.44$)} & $d_0$ & $d_1$ & $d_2$ & $d_3$ & $d_4$ \\ \cline{2-6}
  & $13.36$ & $-263.03$ & $86.00$ & $91.21$ & $102.50$\\
  \bottomrule
\end{tabular}
\end{table}
\begin{figure}[ht!]
	\centering
    \includegraphics[scale=0.6]{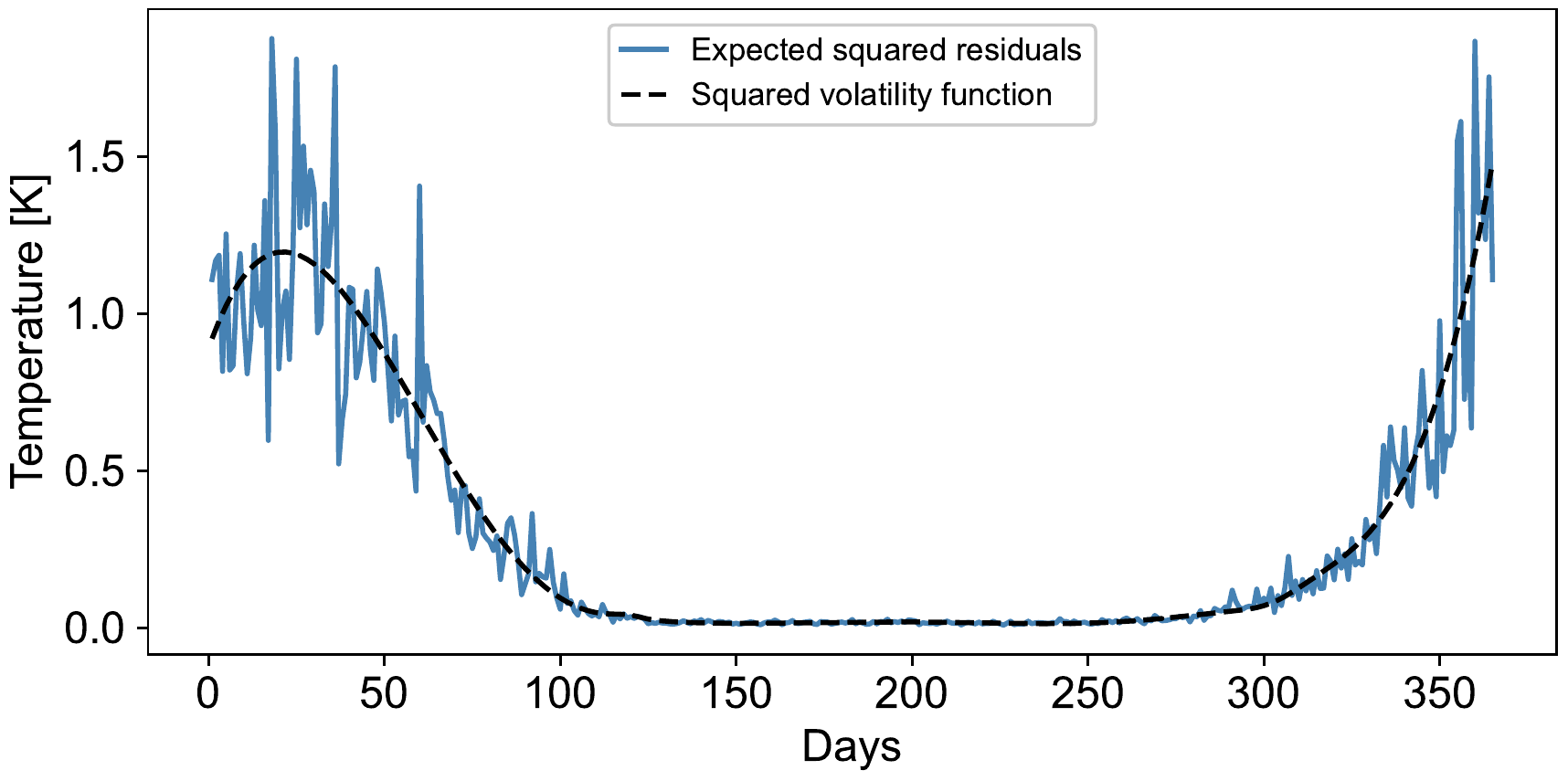}
    \caption{Estimation of expected squared residuals (estimated variance) each day of the year illustrated together with a fitted function}
    \label{fig:expected_sq_res}
\end{figure}
With an explicit expression for the volatility function $\sigma(t)$, the $\sigma(t)$-scaled residuals $\epsilon(t)$ (as described in Eq.\,\eqref{eq:res_with_seasonality_func}) can be studied. The distribution of the $\sigma(t)$-scaled stratospheric temperature model residuals in $\hat{\mathcal{E}}$ is compared to the normal distribution with a QQ-plot in Fig.\,\ref{fig:res_ds_qq}. The above hypothesis about $\epsilon(t)$ being i.i.d.\ $N(0,1)$ random variables does not hold, as the QQ-plot illustrate heavy tails and a slightly skewed distribution. Also, the Kolmogorov-Smirnov (KS) test with statistic $0.027$ and $p$-value $1.57\cdot 10^{-9}$ gives significance for rejecting the hypotheses about $\epsilon(t)$ being standard normal. A fit with the NIG distribution is further performed, and the resulting pdf is illustrated in Fig.\,\ref{fig:res_ds_nig}. The KS test with statistic $0.0064$ and $p$-value $0.57$ does not reject the null-hypothesis that $\epsilon(t)$ represents NIG distributed random variables. Note that the KS-test is meant to provide indicative results, rather than concluding results from a carefully planned statistical experiment. 

\begin{figure}[ht!]
	\centering
	\begin{subfigure}[ht!]{0.49\textwidth}
	\centering	
	\includegraphics[scale=0.6]{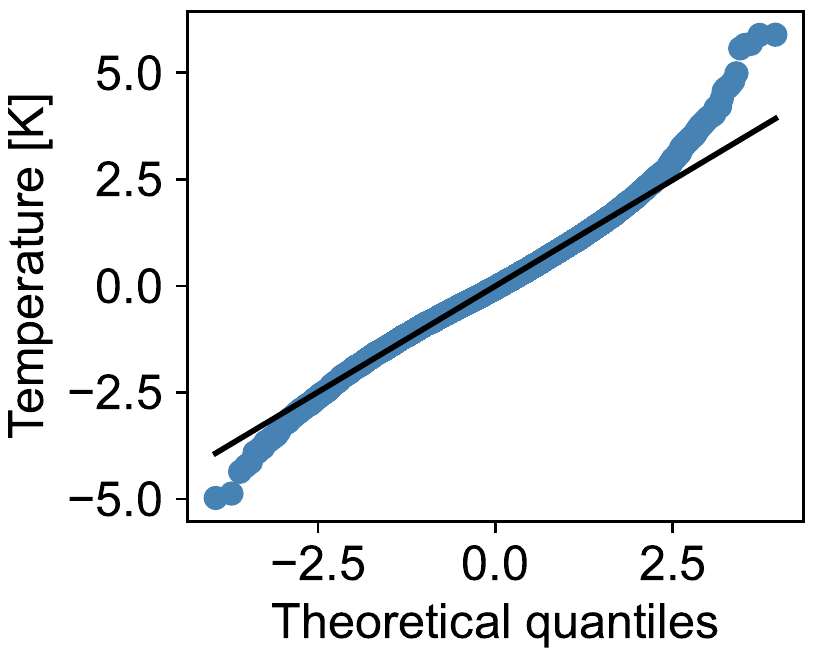}
	\caption{Normal QQ-plot of $\epsilon(t)$}
	\label{fig:res_ds_qq}
	\end{subfigure}
	\begin{subfigure}[ht!]{0.49\textwidth}
	\centering
	\includegraphics[scale=0.6]{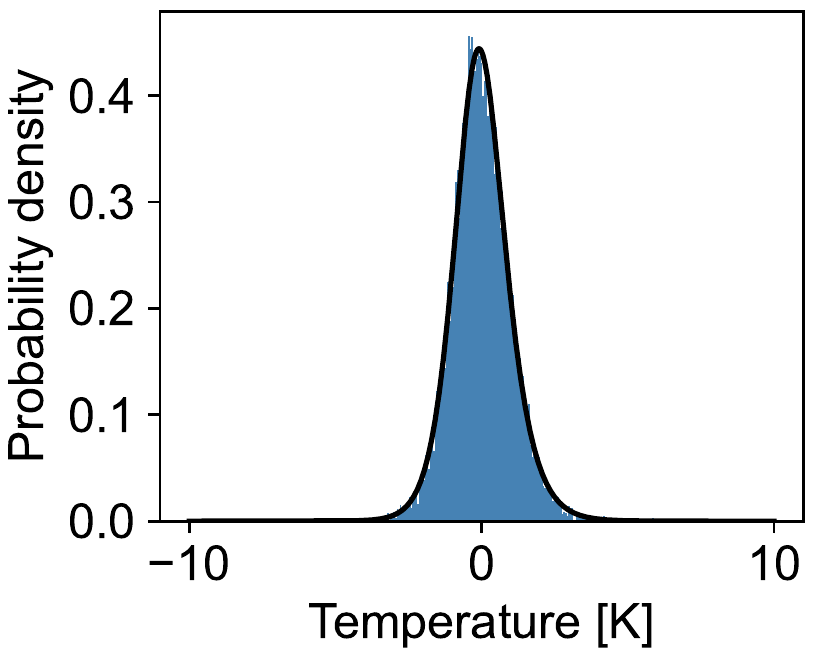}
	\caption{NIG fit of $\epsilon(t)$}
	\label{fig:res_ds_nig}
	\end{subfigure}
    \caption{Fitted distributions to the $\sigma(t)$-scaled stratospheric temperature model residuals $\epsilon(t)$, see Eq.\,\eqref{eq:res_with_seasonality_func}}
    \label{fig:res_ds_fit}
\end{figure}
The result of $\epsilon(t)$ representing NIG distributed random variables supports the hypothesis of using a L{\'e}vy process as the driving process for the stratospheric temperature model, as proposed in Sect.\,\ref{sect:car_ar_connection}. However, as shown in Figs.\,\ref{fig:acf_res_ds} and \ref{fig:pacf_res_ds}, the squared $\sigma(t)$-scaled residuals are partially autocorrelated in approximately $5$ lags, meaning increments of $\epsilon(t)$ fail to be independently distributed. %
These memory effects indicate using a stochastic volatility function as described in \cite{benth09}, rather than a deterministic yearly time-varying volatility function. To generalize the proposed model in Sect.\,\ref{sect:car_ar_connection}, a possibility would be to model $e(t)$ as a normal variance-mean mixture with an inverse Gaussian stochastic volatility, as this process is approximately NIG distributed, see \cite{barndorff-n01} and \cite{benth13}, Sect.\,$4$. %
However, as the memory effects are rather small (except in the first lag), it could be appropriate to assume that there are no significant memory effects in the variance. A possibility is therefore to assume a deterministic volatility function, where the driving process for the stratospheric temperature model is a NIG L{\'e}vy process (e.g.,  \cite{barndorff97-2}, \cite{barndorff97-1} and \cite{barndorff-n01}).
Further studying of this aspect is beyond the scope of the current paper, and is left as a topic for further research.
\begin{figure}[ht!]
	\centering
	\begin{subfigure}[ht!]{0.49\textwidth}
	\centering	
	\includegraphics[scale=0.6]{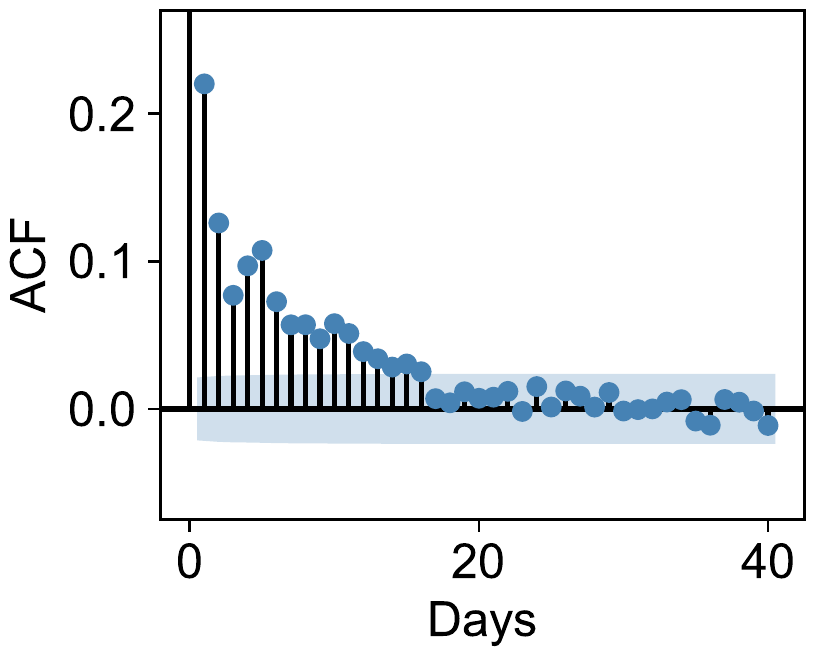}
	\caption{ACF of $\epsilon^2(t)$}
	\label{fig:acf_res_ds}
	\end{subfigure}
	\begin{subfigure}[ht!]{0.49\textwidth}
	\centering
	\includegraphics[scale=0.6]{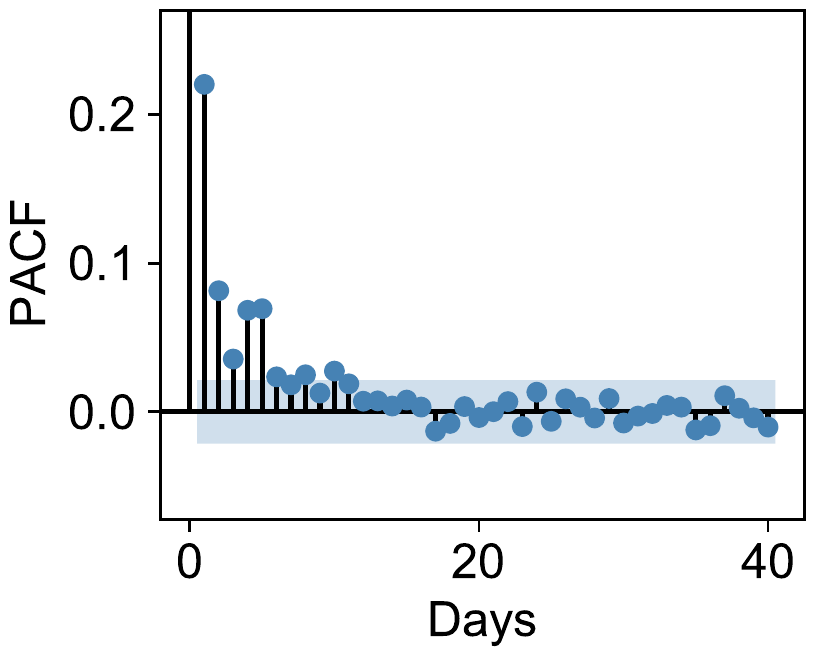}
	\caption{PACF of $\epsilon^2(t)$}
	\label{fig:pacf_res_ds}
	\end{subfigure}
    \caption{ACF and PACF of the squared $\sigma(t)$-scaled stratospheric temperature model residuals $\epsilon^2(t)$, see Eq.\,\eqref{eq:res_with_seasonality_func}}
    \label{fig:acf_pacf_res_ds}
\end{figure}


\section{\label{sec:mean_reversion}Analyzing the speed of mean reversion}
In this section, it is shown that the assumption of constant speed of mean reversion for stratospheric temperature is erroneous. A generalization of the proposed stratospheric temperature model dynamics in Eq.\,\eqref{eq:car(p)-model}, correcting this erroneous assumption, is presented. More specifically, a special case of the dynamics in Eq.\,\eqref{eq:CAR_t_dep_speed} is proposed as a replacement to drive the stratospheric temperature model in Eq.\,\eqref{eq:strat_temp_model}. 


\subsection{\label{subsec: mean_reversion_method}Methodology for analyzing speed of mean reversion}

In the previous sections, it was shown that the deseasonalized stratospheric temperature, $Y(t)$, follows a mean-reverting stochastic process. That is, deseasonalized stratospheric temperature dynamics is given by the OU process in Eq.\,\eqref{eq:car(p)-model}. The matrix $A$ holds parameters of speed of mean reversion, meaning that the rate at which the stratospheric temperature reverts back to its long-term mean is given by the elements of $A$, see Sect.\,\ref{sec:autoreg_behav} (Eq.\,\eqref{eq: formula} and Eq.\,\eqref{eq:CAR_speed}). In the previous sections the speed of mean reversion was assumed to be constant, and thus independent of time. To check the validity of this assumption, a similar stability analysis of speed of mean reversion as in \cite{benth05} will be performed in the following. The stability analysis exploits the transformation relation between CAR and AR models (see Sect.\,\ref{sect:car_ar_connection}), meaning that it is applied on computed AR parameters $\beta (t)$. That is: The mean value, $E[\beta(t)]$, and standard deviation, $\sqrt{\text{Var}(\beta(t))}$, of fitted AR($p$) parameters are computed empirically over each available year and month. Based on this, the yearly and monthly variation coefficients are found as 
\begin{align}
    \label{eq:var_coeff}
    \Delta = \frac{\sqrt{\text{Var}(\beta(t))}}{E[\beta(t)]}.
\end{align}
The yearly and monthly variation coefficients reflect the stability of the speed of mean reversion over years and months, respectively. As the stratospheric temperature model is derived with lags in four days, this analysis will be performed with $p=4$ for all computed AR parameters. The methodology for analyzing the yearly stability of speed of mean reversion is as follows: 1; Loop through the $40$ years in $\mathcal{S}_d$ and collect deseasonalized daily-zonal mean stratospheric temperatures $365$ days at a time, such that data for 1 January 1979 to 31 December 1979 are collected in one array, data for 1 January 1980 to 31 December 1980 in one array, and so on until the last array containing data for 1 January 2018 to 31 December 2018. Then, 2; Collect the $40$ arrays containing all data each year in a single array to form the nested array $\boldsymbol{y}$, so $\text{dim}(\boldsymbol{y}) = 1\times 40 \times 365$. 3; Loop through the $40$ arrays in $\boldsymbol{y}$, where an AR($4$) model is fit to each of the years 1979 to 2018. The result is $40$ arrays of AR parameters $(\beta_1^y,\beta_2^y,\beta_3^y,\beta_4^y)$, $y\in\{1,2,\ldots ,40\}$. 4; Collect all $40$ AR parameters corresponding to the same lag in one array and compute the statistics. That is, make the arrays $\boldsymbol{\beta}_1^y = (\beta_1^1, \beta_1^2,\ldots,\beta_1^{40}),\ldots,\boldsymbol{\beta}_4^y = (\beta_4^1, \beta_4^2,\ldots,\beta_4^{40})$, and compute the empirical mean, standard deviation, and finally the variation coefficients defined in Eq.\,\eqref{eq:var_coeff}, for each array $\boldsymbol{\beta}_1^y$, $\boldsymbol{\beta}_2^y$, $\boldsymbol{\beta}_3^y$ and $\boldsymbol{\beta}_4^y$.
The results from this yearly stability analysis are presented in Tab.\,\ref{tab:stab_mean_rev_y_and_m}. Further, the methodology for the monthly stability analysis of speed of mean reversion is: 1; Define one array for each month: $\boldsymbol{Jan},\boldsymbol{Feb},\boldsymbol{Mar},\ldots,\boldsymbol{Dec}$. 2; Loop through the $40$ arrays in $\boldsymbol{y}$ (which is constructed in point 2 for the yearly stability analysis). For all $40$ arrays, collect elements $0$ to $30$ in $\boldsymbol{Jan}$, elements $31$ to $58$ in $\boldsymbol{Feb}$, elements $59$ to $89$ in $\boldsymbol{Mar}$, and so on until you reach elements $334$ to $364$ which are collected in $\boldsymbol{Dec}$. The resulting arrays are nested arrays of the form
    \begin{align*}
        \boldsymbol{Jan} = \left[
        \boldsymbol{Jan}^{1}, \ldots,\boldsymbol{Jan}^{40}
        \right], \boldsymbol{Feb}= \left[\boldsymbol{Feb}^{1}, \ldots, \boldsymbol{Feb}^{40}\right],   \ldots,  \boldsymbol{Dec} = \left[\boldsymbol{Dec}^{1}, \ldots, \boldsymbol{Dec}^{40}\right].
    \end{align*}
Each array has dimension $1 \times 40 \times n$, where $n$ corresponds to the number of days in that particular month. That is, $n=31$ for $\boldsymbol{Jan}$, $n=28$ for $\boldsymbol{Feb}$ and so on. 3; Collect the arrays $\boldsymbol{Jan},\ldots,\boldsymbol{Dec}$ in a nested array $\boldsymbol{m}$: $\text{dim}(\boldsymbol{m}) = 1\times 12 \times 40 \times n$. Loop through each of the $480$ months in $\boldsymbol{m}$, and fit an AR($4$) model to each of the Januaries of the years 1979 to 2018, to each of the Februaries of the years 1979 to 2018, and so on until the last fit is performed on the data of December of 2018. The result is $480$ arrays of AR parameters $(\beta_1^m,\beta_2^m,\beta_3^m,\beta_4^m)$, $m\in\{1,2,\ldots ,480\}$. 4; Make the arrays $\boldsymbol{\beta}_1^m = (\beta_1^1, \beta_1^2,\ldots,\beta_1^{480}),\ldots,\boldsymbol{\beta}_4^m = (\beta_4^1, \beta_4^2,\ldots,\beta_4^{480})$, and for each of them compute the variation coefficient as defined in Eq.\,\eqref{eq:var_coeff}. Further, because of the constructed order of the parameters, the arrays $\boldsymbol{\beta}_1^m$, $\boldsymbol{\beta}_2^m$, $\boldsymbol{\beta}_3^m$ and $\boldsymbol{\beta}_4^m$ can be used to analyze the seasonal behaviour of speed of mean reversion.
The computed results from this monthly stability analysis are presented in Tab.\,\ref{tab:stab_mean_rev_y_and_m}, where the variability coefficients will reveal any monthly instability of speed of mean reversion. With the intention of detecting any monthly seasonal behaviour in the four AR parameters, $\boldsymbol{\beta}_1^m$, $\boldsymbol{\beta}_2^m$, $\boldsymbol{\beta}_3^m$ and $\boldsymbol{\beta}_4^m$ are plotted in Figs.\,\ref{fig:beta_1}-\ref{fig:beta_4}.


\subsection{\label{subsec: mean_reversion_interpret}Interpretations of the monthly stability analysis}
The monthly variation coefficients are more extreme than the yearly ones. Therefore, the remaining of this section will focus on interpreting results from the monthly stability analysis, as well as to incorporate the observed time-varying behaviour of the speed of mean reversion, into the stratospheric temperature model dynamics in Eq.\,\eqref{eq:CAR_t_dep_speed}.

The magnitudes of the monthly variation coefficients (see Tab.\,\ref{tab:stab_mean_rev_y_and_m}) for all four AR parameters, indicate that the assumption of constant speed of mean reversion in the stratospheric temperature model dynamics (Eq.\,\eqref{eq:car(p)-model}) is insufficient. The monthly variation coefficient, $\Delta$, of the first lag parameter, $\beta_1$, is small compared to $\Delta$ of the three other lag parameters. 
However, as seen in Figs.\,\ref{fig:beta_1}-\ref{fig:beta_4}, the magnitude of $\beta_1$ is up to many times larger than the magnitudes of $\beta_2$, $\beta_3$ and $\beta_4$ (where the magnitudes of $\beta_1$, $\beta_2$, $\beta_3$ and $\beta_4$ are assessed over the contents of $\boldsymbol{\beta}_1^m$, $\boldsymbol{\beta}_2^m$, $\boldsymbol{\beta}_3^m$ and $\boldsymbol{\beta}_4^m$ respectively). %
This means that larger variability in the latter lag parameters affect the estimated stratospheric temperature less. Despite this observation, all variability coefficients are too large to be ignored, meaning that time-dependent AR parameters should be used rather than constant ones. By the transformation relation between CAR($4$) and AR($4$) models in Eq.\,\eqref{eq:ar_car_connection}, it is clear that the stratospheric temperature model dynamics should be given by the OU process in Eq.\,\eqref{eq:CAR_t_dep_speed}, rather than the one in Eq.\,\eqref{eq:car(p)-model}.

It is not only the monthly variability coefficients of the AR parameters that suggest time-varying speed of mean reversion. The sequential patterns of $\boldsymbol{\beta}_1^m$ and $\boldsymbol{\beta}_2^m$ in Figs.\,\ref{fig:beta_1} and \ref{fig:beta_2} clearly show that the AR parameters $\beta_1$ and $\beta_2$ are seasonally varying. Both the first, $\beta_1$, and the second, $\beta_2$, AR parameters are smaller in magnitude in summer time than in winter time. This means that summer time stratospheric temperature is less dependent on the stratospheric temperature the last two days, than winter time stratospheric temperature. This tendency is also (however less) evident for the third, $\beta_3$, and fourth, $\beta_4$, AR parameters in Figs.\,\ref{fig:beta_3} and \ref{fig:beta_4}.
\begin{table}[ht!]
    \caption{Mean value, standard deviation and absolute value of variability coefficient for the four parameters of an AR($4$) process with yearly and monthly varying parameters, respectively}
    \label{tab:stab_mean_rev_y_and_m}
    \centering
    \begin{tabular}{ c c c c c c }
  \toprule
   & Parameter & $\beta_1$ & $\beta_2$ & $\beta_3$ & $\beta_4$\\ \cline{2-6}
  \multirow{2.0}{*}{Yearly} & Mean value & $1.54$ & $-0.74$ & $0.28$ & $-0.12$\\
  & Standard deviation & $0.12$ & $0.22$ & $0.18$ & $ 0.09$\\
  & $\text{abs}(\Delta)$ & $7.8$\% & $29.9$\% & $64.3$\% & $72.6$\%\\
  \bottomrule
  \toprule
  & Parameter & $\beta_1$ & $\beta_2$ & $\beta_3$ & $\beta_4$\\ \cline{2-6}
  \multirow{2.0}{*}{Monthly} & Mean value & $1.29$ & $-0.48$ & $0.19$ & $-0.06$\\
  & Standard deviation & $0.34$ & $0.45$ & $0.32$ & $0.20$\\
  & $\text{abs}(\Delta)$ & $26.6$\% & $94.8$\% & $170.7$\% & $312.9$\%\\
  \bottomrule
\end{tabular}
\end{table}
To the best of our knowledge there is no previous research specifically on the mean reverting property of stratospheric temperature. In \cite{zapranis08}, daily values of speed of mean reversion for surface temperature are estimated by use of a neural network, revealing strong time-dependence. Even though daily variation in speed of mean reversion is not studied in the current paper, a similar conclusion is reached: Speed of mean reversion of stratospheric temperature is dependent on time. However, \cite{zapranis08} found no signs of seasonal patterns, unlike the current study for stratospheric temperature where a clear seasonal pattern is observed in the monthly estimated AR parameters. In response to the observation of time-dependence in speed of mean reversion of surface temperature, \cite{benth15} presented a generalized version of the state-of-the-art stochastic models for surface temperatures applied in mathematical finance (see Sect.\,\ref{sect:car_ar_connection}), where the speed of mean reversion of the driving (standard) OU process is a stochastic process. A simplified version of this generalized model (however L{\'e}vy-driven rather than Brownian motion-driven) is presented in the current paper to incorporate time variability in speed of mean reversion. That is, the matrix $A(t)$ holding parameters of speed of mean reversion is assumed to be time-dependent and deterministic as presented in Thm.\,\ref{thm:vectorial_OU_4}. 

Before presenting an explicit time-dependent and deterministic matrix $A(t)$ representing the speed of mean reversion of stratospheric temperature, its time-varying behaviour will be discussed. As already mentioned, the monthly AR parameters (which the speed of mean reversion depends directly upon) have large variability coefficients, $\Delta$, as well as a seasonal behaviour. From the definition of $\Delta$ in Eq.\,\eqref{eq:var_coeff}, one can see that the seasonal behaviour increases the computed variability coefficients considerably. By removing the seasonal behaviour in the AR parameters (see below) noise is still present, indicating that the speed of mean reversion could be modeled by a stochastic process. However, this is beyond the scope of this paper, and the noise is assumed to be negligible. For this reason, time-dependence in speed of mean reversion is assumed to come solely from the seasonal variations. 

As discussed in Sect.\,\ref{sec:season_function_data}, the only long-term (perfectly) periodic phenomenon in the stratosphere is the yearly cycle. This is clearly seen in the stratospheric temperature data presented in Fig.\,\ref{fig:temp_10y_with_seasonfit}. Further, as seen in Fig.\,\ref{fig:expected_sq_res}, this phenomenon affects the variability, as well as volatility in variability, of stratospheric temperature. Physical explanations of this behaviour are discussed in \cite{haynes2005}, where the winter time stratosphere is said to be more disturbed than the summer time stratosphere. Based on this, together with the above discussion concluding stronger speed of mean reversion in winter time than in summer time, it might be reasonable to assume that the yearly cycle affects the speed of mean reversion of stratospheric temperature as well. Stated in another way: Large values of stratospheric temperature variance seem to generate larger dependence on stratospheric temperature the last couple of days. This is a topic for further research. 

An explicit deterministic matrix $A(t)$ representing speed of mean reversion of stratospheric temperature is proposed in the following. By introducing time-varying AR($4$) parameters such that each month of the year holds fixed parameters, seasonal variability will be adjusted for on a monthly basis. The result of the monthly variation analysis is exploited to define such time varying AR($4$) parameters. That is, define monthly parameter values as (remember that the initial dataset $\mathcal{S}$ contains stratospheric temperature values over $40$ years) 
\begin{align*}
    \beta_{k}^M \triangleq 
    \begin{cases}
    \beta_{k}^1 \\
    \beta_{k}^2 \\
    \vdots \\
    \beta_{k}^{12}
    \end{cases} =
    \begin{cases}
    E[\boldsymbol{\beta}_k^m],\quad \text{for } m\in [1,40] \\
    E[\boldsymbol{\beta}_k^m],\quad \text{for } m\in [41,80] \\
    \quad\vdots \\
    E[\boldsymbol{\beta}_k^m],\quad \text{for } m\in [441,480], 
    \end{cases}
\end{align*}
where $M\in\{1,2,\ldots,12\}$ represents January to December respectively, $k\in\{1,2,3,4\}$ and $E[\cdot]$ represents the empirical mean. The computed values of $\beta_1^M$, $\beta_2^M$, $\beta_3^M$ and $\beta_4^M$ are marked in Figs.\,\ref{fig:beta_1}-\ref{fig:beta_4}, respectively. Building on the theory in Sect.\,\ref{sect:car_ar_connection}, the corresponding CAR($4$) model is given by the multidimensional OU process in Eq.\,\eqref{eq:CAR_t_dep_speed} (Thm.\,\ref{thm:vectorial_OU_4}). Time-dependence in the functions $\alpha_1(t), \ldots \alpha_4(t)$ does not matter for the transformation relation between CAR and AR models as long as the discretization scheme is chosen properly, see \cite{iacus08} and \cite{benth2008}. That is, in this specific case, the discretization scheme has to be constructed such that the two time points which define the current scheme time step, never belongs to two different months. Hence, the continuous parameter counterparts, $\alpha_k(t)$, of the $\beta_k^M$'s can be computed by the transformation relation in Eq.\,\eqref{eq:ar_car_connection}. The $\alpha_k(t)$'s can be considered as $12$ level step functions, where each step represents a month of the year. The roots of the eigenvalue equation (Eq.\,\eqref{eq:eigenvalueeq_ar(4)}) of each of the $12$ steps in each of the functions $\alpha_k(t)$ are computed, and their real parts are shown in Fig.\,\ref{fig:roots}. As all the roots have negative real part the stationarity condition of CAR processes is secured. 
\begin{figure}[ht!]
	\centering
	\begin{subfigure}[ht!]{0.49\textwidth}
	\centering	
	\includegraphics[scale=0.6]{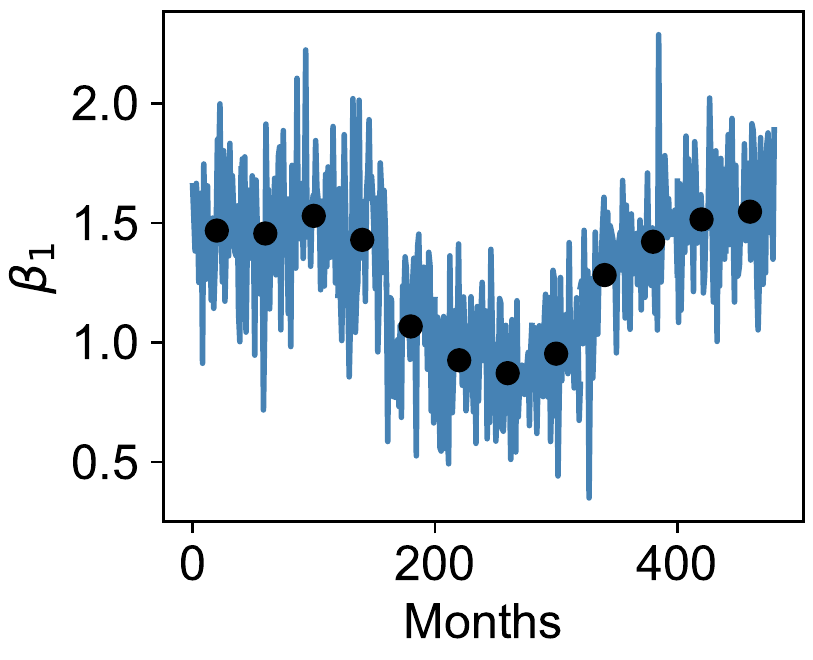}
	\caption{Monthly variability in $\beta_1$}
	\label{fig:beta_1}
	\end{subfigure}
	\begin{subfigure}[ht!]{0.49\textwidth}
	\centering
	\includegraphics[scale=0.6]{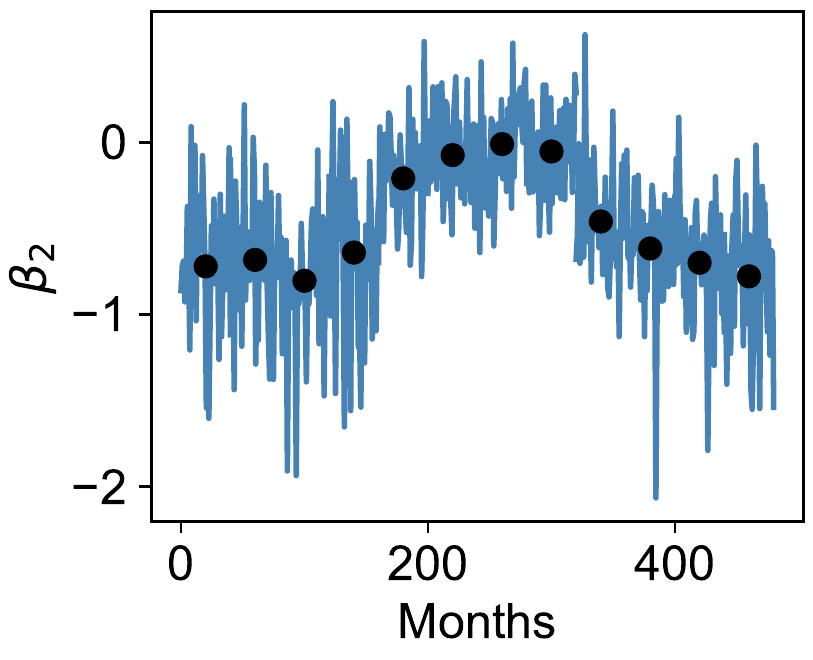}
	\caption{Monthly variability in $\beta_2$}
	\label{fig:beta_2}
	\end{subfigure}
	\begin{subfigure}[ht!]{0.49\textwidth}
	\centering	
	\includegraphics[scale=0.6]{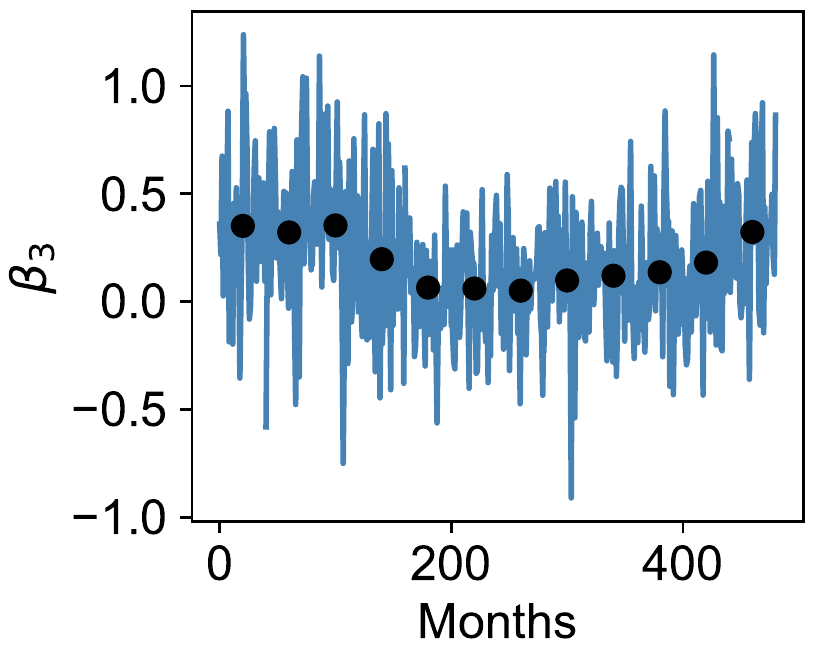}
	\caption{Monthly variability in $\beta_3$}
	\label{fig:beta_3}
	\end{subfigure}
	\begin{subfigure}[ht!]{0.49\textwidth}
	\centering
	\includegraphics[scale=0.6]{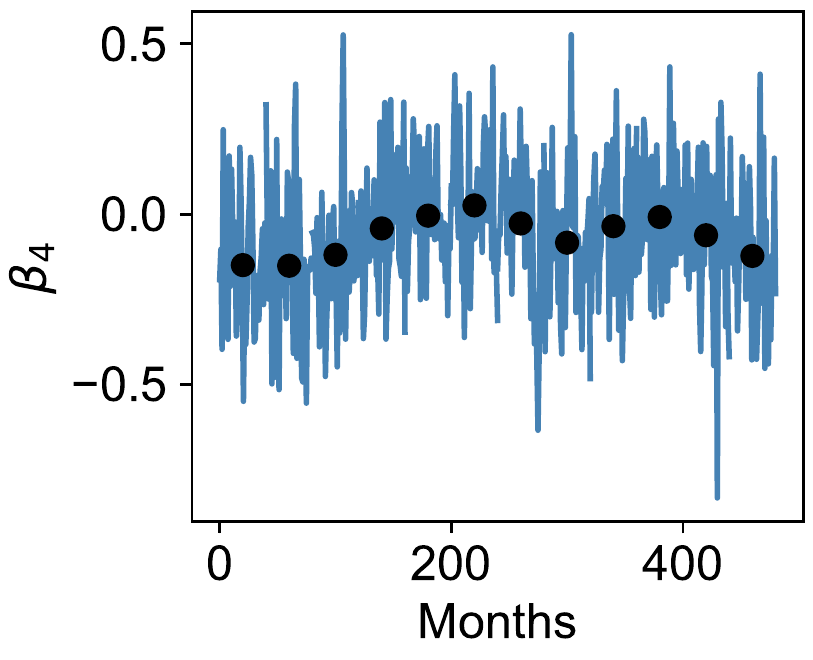}
	\caption{Monthly variability in $\beta_4$}
	\label{fig:beta_4}
	\end{subfigure}
    \caption{Monthly variability in the four AR parameters of the stratospheric temperature model}
    \label{fig:stability_ar}
\end{figure}
\begin{figure}[ht!]
	\centering
	\includegraphics[scale=0.6]{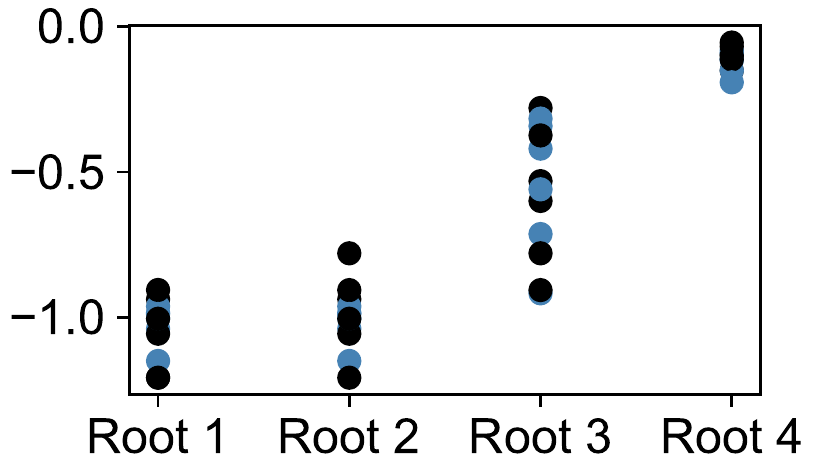}
	\caption{Real part of the roots of the eigenvalue equations (Eq.\,\eqref{eq:eigenvalueeq_ar(4)}) of the monthly varying CAR($4$) parameters}
	\label{fig:roots}
\end{figure}
Repeating the analysis in Sect.\,\ref{sect:residuals_analysis} with model residuals
\begin{align*}
    e(t) = X(t) - \sum_{k=1}^4 \beta_{k}^MX(t-k),
\end{align*}
also gives residuals on the form $e(t) = \sigma(t)\epsilon(t)$, with $\epsilon(t)$ being NIG distributed random variables with memory effects. This confirms that the proposed CAR($4$) model from Sect.\,\ref{sect:car_ar_connection} is suitable to model stratospheric temperature dynamics when driven by the multidimensional OU dynamics in Eq.\,\eqref{eq:CAR_t_dep_speed}, see Thm.\,\ref{thm:vectorial_OU_4}. 


\section{\label{sec: conclusions}Conclusions and further work}
In this paper, a novel stochastic model for stratospheric temperature is proposed. By time series analysis, it was shown that stratospheric temperature can be approximated by an AR(4) process added to a deterministic seasonality function. The seasonality function captures periodical (yearly seasonal) effects as well as a long-term trend to model stratospheric cooling. The scaled model residuals were shown to be NIG distributed. By exploiting the connection between AR and CAR processes, a continuous time model for stratospheric temperature was developed. It was shown that a L{\'e}vy driven CAR(4) process with time-dependent volatility is well suited as a continuous time, stochastic model for deseasonalized stratospheric temperature.

Some ideas for future work include incorporating a volatility which is stochastic, not just time-dependent, into the model. Furthermore, it may be of interest to analyze further why large values of variance in the stratospheric temperature seem to generate a larger dependency on stratospheric temperature the previous days. Developing a continuous time, stochastic model for stratospheric wind, potentially as a joint model with stratospheric temperature, is also relevant. In addition, based on the model presented in the current paper, one may exploit stratosphere-troposphere coupling in order to develop improved methods for pricing of weather derivatives (on surface-level). A current work in progress is to develop a dual model for stratospheric temperature, where winter season and summer season temperatures are studied separately. Such a model is particularly useful when analyzing, for example, pure winter phenomena, such as sudden stratospheric warmings.

\bibliography{bibl}

\begin{thebibliography}{40}%
\makeatletter
\providecommand \@ifxundefined [1]{%
 \@ifx{#1\undefined}
}%
\providecommand \@ifnum [1]{%
 \ifnum #1\expandafter \@firstoftwo
 \else \expandafter \@secondoftwo
 \fi
}%
\providecommand \@ifx [1]{%
 \ifx #1\expandafter \@firstoftwo
 \else \expandafter \@secondoftwo
 \fi
}%
\providecommand \natexlab [1]{#1}%
\providecommand \enquote  [1]{``#1''}%
\providecommand \bibnamefont  [1]{#1}%
\providecommand \bibfnamefont [1]{#1}%
\providecommand \citenamefont [1]{#1}%
\providecommand \href@noop [0]{\@secondoftwo}%
\providecommand \href [0]{\begingroup \@sanitize@url \@href}%
\providecommand \@href[1]{\@@startlink{#1}\@@href}%
\providecommand \@@href[1]{\endgroup#1\@@endlink}%
\providecommand \@sanitize@url [0]{\catcode `\\12\catcode `\$12\catcode
  `\&12\catcode `\#12\catcode `\^12\catcode `\_12\catcode `\%12\relax}%
\providecommand \@@startlink[1]{}%
\providecommand \@@endlink[0]{}%
\providecommand \url  [0]{\begingroup\@sanitize@url \@url }%
\providecommand \@url [1]{\endgroup\@href {#1}{\urlprefix }}%
\providecommand \urlprefix  [0]{URL }%
\providecommand \Eprint [0]{\href }%
\providecommand \doibase [0]{https://doi.org/}%
\providecommand \selectlanguage [0]{\@gobble}%
\providecommand \bibinfo  [0]{\@secondoftwo}%
\providecommand \bibfield  [0]{\@secondoftwo}%
\providecommand \translation [1]{[#1]}%
\providecommand \BibitemOpen [0]{}%
\providecommand \bibitemStop [0]{}%
\providecommand \bibitemNoStop [0]{.\EOS\space}%
\providecommand \EOS [0]{\spacefactor3000\relax}%
\providecommand \BibitemShut  [1]{\csname bibitem#1\endcsname}%
\let\auto@bib@innerbib\@empty
\bibitem [{\citenamefont {Baldwin}\ \emph {et~al.}(2021)\citenamefont
  {Baldwin}, \citenamefont {Ayarzag{\"u}ena}, \citenamefont {Birner},
  \citenamefont {Butchart}, \citenamefont {Butler}, \citenamefont
  {Charlton-Perez}, \citenamefont {Domeisen}, \citenamefont {Garfinkel},
  \citenamefont {Garny}, \citenamefont {Gerber}, \citenamefont {Hegglin},
  \citenamefont {Langematz},\ and\ \citenamefont
  {Pedatella}}]{baldwin2021sudden}%
  \BibitemOpen
  \bibfield  {author} {\bibinfo {author} {\bibnamefont {Baldwin}, \bibfnamefont
  {M.~P.}}, \bibinfo {author} {\bibnamefont {Ayarzag{\"u}ena}, \bibfnamefont
  {B.}}, \bibinfo {author} {\bibnamefont {Birner}, \bibfnamefont {T.}},
  \bibinfo {author} {\bibnamefont {Butchart}, \bibfnamefont {N.}}, \bibinfo
  {author} {\bibnamefont {Butler}, \bibfnamefont {A.~H.}}, \bibinfo {author}
  {\bibnamefont {Charlton-Perez}, \bibfnamefont {A.~J.}}, \bibinfo {author}
  {\bibnamefont {Domeisen}, \bibfnamefont {D.~I.~V.}}, \bibinfo {author}
  {\bibnamefont {Garfinkel}, \bibfnamefont {C.~I.}}, \bibinfo {author}
  {\bibnamefont {Garny}, \bibfnamefont {H.}}, \bibinfo {author} {\bibnamefont
  {Gerber}, \bibfnamefont {E.~P.}}, \bibinfo {author} {\bibnamefont {Hegglin},
  \bibfnamefont {M.~I.}}, \bibinfo {author} {\bibnamefont {Langematz},
  \bibfnamefont {U.}}, and\ \bibinfo {author} {\bibnamefont {Pedatella},
  \bibfnamefont {N.~M.}},\ }\bibfield  {title} {\enquote {\bibinfo {title}
  {Sudden stratospheric warmings},}\ }\href
  {https://doi.org/https://doi-org.ezproxy.uio.no/10.1029/2020RG000708}
  {\bibfield  {journal} {\bibinfo  {journal} {Reviews of Geophysics}\ }\textbf
  {\bibinfo {volume} {59}} (\bibinfo {year} {2021}),\
  https://doi-org.ezproxy.uio.no/10.1029/2020RG000708}\BibitemShut {NoStop}%
\bibitem [{\citenamefont {Baldwin}\ \emph {et~al.}(2019)\citenamefont
  {Baldwin}, \citenamefont {Birner}, \citenamefont {Brasseur}, \citenamefont
  {Burrows}, \citenamefont {Butchart}, \citenamefont {Garcia}, \citenamefont
  {Geller}, \citenamefont {Gray}, \citenamefont {Hamilton}, \citenamefont
  {Harnik}, \citenamefont {Hegglin}, \citenamefont {Langematz}, \citenamefont
  {Robock}, \citenamefont {Sato},\ and\ \citenamefont
  {Scaife}}]{baldwin2019_100years}%
  \BibitemOpen
  \bibfield  {author} {\bibinfo {author} {\bibnamefont {Baldwin}, \bibfnamefont
  {M.~P.}}, \bibinfo {author} {\bibnamefont {Birner}, \bibfnamefont {T.}},
  \bibinfo {author} {\bibnamefont {Brasseur}, \bibfnamefont {G.}}, \bibinfo
  {author} {\bibnamefont {Burrows}, \bibfnamefont {J.}}, \bibinfo {author}
  {\bibnamefont {Butchart}, \bibfnamefont {N.}}, \bibinfo {author}
  {\bibnamefont {Garcia}, \bibfnamefont {R.}}, \bibinfo {author} {\bibnamefont
  {Geller}, \bibfnamefont {M.}}, \bibinfo {author} {\bibnamefont {Gray},
  \bibfnamefont {L.}}, \bibinfo {author} {\bibnamefont {Hamilton},
  \bibfnamefont {K.}}, \bibinfo {author} {\bibnamefont {Harnik}, \bibfnamefont
  {N.}}, \bibinfo {author} {\bibnamefont {Hegglin}, \bibfnamefont {M.~I.}},
  \bibinfo {author} {\bibnamefont {Langematz}, \bibfnamefont {U.}}, \bibinfo
  {author} {\bibnamefont {Robock}, \bibfnamefont {A.}}, \bibinfo {author}
  {\bibnamefont {Sato}, \bibfnamefont {K.}}, and\ \bibinfo {author}
  {\bibnamefont {Scaife}, \bibfnamefont {A.~A.}},\ }\bibfield  {title}
  {\enquote {\bibinfo {title} {100 years of progress in understanding the
  stratosphere and mesosphere},}\ }\href
  {https://doi.org/https://doi.org/10.1175/AMSMONOGRAPHS-D-19-0003.1}
  {\bibfield  {journal} {\bibinfo  {journal} {Meteorological Monographs}\
  }\textbf {\bibinfo {volume} {59}},\ \bibinfo {pages} {27.1–--27.62}
  (\bibinfo {year} {2019})}\BibitemShut {NoStop}%
\bibitem [{\citenamefont {Baldwin}\ and\ \citenamefont
  {Dunkerton}(2001)}]{BaldwinDunkerton2001}%
  \BibitemOpen
  \bibfield  {author} {\bibinfo {author} {\bibnamefont {Baldwin}, \bibfnamefont
  {M.~P.}}and\ \bibinfo {author} {\bibnamefont {Dunkerton}, \bibfnamefont
  {T.~J.}},\ }\bibfield  {title} {\enquote {\bibinfo {title} {Stratospheric
  harbingers of anomalous weather regimes},}\ }\href
  {https://doi.org/https://doi.org/10.1126/science.1063315} {\bibfield
  {journal} {\bibinfo  {journal} {Science}\ }\textbf {\bibinfo {volume}
  {294}},\ \bibinfo {pages} {581--584} (\bibinfo {year} {2001})}\BibitemShut
  {NoStop}%
\bibitem [{\citenamefont {Baldwin}\ \emph {et~al.}(2001)\citenamefont
  {Baldwin}, \citenamefont {Gray}, \citenamefont {Dunkerton}, \citenamefont
  {Hamilton}, \citenamefont {Haynes}, \citenamefont {Randel}, \citenamefont
  {Holton}, \citenamefont {Alexander}, \citenamefont {Hirota}, \citenamefont
  {Horinouchi}, \citenamefont {Jones}, \citenamefont {Kinnersley},
  \citenamefont {Marquardt}, \citenamefont {Sato},\ and\ \citenamefont
  {Takahashi}}]{baldwin2001}%
  \BibitemOpen
  \bibfield  {author} {\bibinfo {author} {\bibnamefont {Baldwin}, \bibfnamefont
  {M.~P.}}, \bibinfo {author} {\bibnamefont {Gray}, \bibfnamefont {L.~J.}},
  \bibinfo {author} {\bibnamefont {Dunkerton}, \bibfnamefont {T.~J.}}, \bibinfo
  {author} {\bibnamefont {Hamilton}, \bibfnamefont {K.}}, \bibinfo {author}
  {\bibnamefont {Haynes}, \bibfnamefont {P.~H.}}, \bibinfo {author}
  {\bibnamefont {Randel}, \bibfnamefont {W.~J.}}, \bibinfo {author}
  {\bibnamefont {Holton}, \bibfnamefont {J.~R.}}, \bibinfo {author}
  {\bibnamefont {Alexander}, \bibfnamefont {M.~J.}}, \bibinfo {author}
  {\bibnamefont {Hirota}, \bibfnamefont {I.}}, \bibinfo {author} {\bibnamefont
  {Horinouchi}, \bibfnamefont {T.}}, \bibinfo {author} {\bibnamefont {Jones},
  \bibfnamefont {D.~B.~A.}}, \bibinfo {author} {\bibnamefont {Kinnersley},
  \bibfnamefont {J.~S.}}, \bibinfo {author} {\bibnamefont {Marquardt},
  \bibfnamefont {C.}}, \bibinfo {author} {\bibnamefont {Sato}, \bibfnamefont
  {K.}}, and\ \bibinfo {author} {\bibnamefont {Takahashi}, \bibfnamefont
  {M.}},\ }\bibfield  {title} {\enquote {\bibinfo {title} {The quasi‐biennial
  oscillation},}\ }\href {https://doi.org/https://doi.org/10.1029/1999RG000073}
  {\bibfield  {journal} {\bibinfo  {journal} {Reviews of Geophysics}\ }\textbf
  {\bibinfo {volume} {39}},\ \bibinfo {pages} {179--229} (\bibinfo {year}
  {2001})}\BibitemShut {NoStop}%
\bibitem [{\citenamefont
  {Barndorff-Nielsen}(1997{\natexlab{a}})}]{barndorff97-2}%
  \BibitemOpen
  \bibfield  {author} {\bibinfo {author} {\bibnamefont {Barndorff-Nielsen},
  \bibfnamefont {O.~E.}},\ }\bibfield  {title} {\enquote {\bibinfo {title}
  {Normal inverse {Gaussian} distributions and stochastic volatility
  modelling},}\ }\href@noop {} {\bibfield  {journal} {\bibinfo  {journal}
  {Scandinavian Journal of Statistics}\ }\textbf {\bibinfo {volume} {24}},\
  \bibinfo {pages} {1--13} (\bibinfo {year} {1997}{\natexlab{a}})}\BibitemShut
  {NoStop}%
\bibitem [{\citenamefont
  {Barndorff-Nielsen}(1997{\natexlab{b}})}]{barndorff97-1}%
  \BibitemOpen
  \bibfield  {author} {\bibinfo {author} {\bibnamefont {Barndorff-Nielsen},
  \bibfnamefont {O.~E.}},\ }\bibfield  {title} {\enquote {\bibinfo {title}
  {Processes of normal inverse {Gaussian} type},}\ }\href
  {https://doi.org/https://doi.org/10.1007/s007800050032} {\bibfield  {journal}
  {\bibinfo  {journal} {Finance and Stochastics}\ }\textbf {\bibinfo {volume}
  {2}},\ \bibinfo {pages} {41--68} (\bibinfo {year}
  {1997}{\natexlab{b}})}\BibitemShut {NoStop}%
\bibitem [{\citenamefont {Barndorff-Nielsen}\ and\ \citenamefont
  {Shephard}(2001)}]{barndorff-n01}%
  \BibitemOpen
  \bibfield  {author} {\bibinfo {author} {\bibnamefont {Barndorff-Nielsen},
  \bibfnamefont {O.~E.}}and\ \bibinfo {author} {\bibnamefont {Shephard},
  \bibfnamefont {N.}},\ }\bibfield  {title} {\enquote {\bibinfo {title}
  {Non-{Gaussian} {Ornstein}-{Uhlenbeck}-based models and some of their uses in
  financial economics},}\ }\href
  {https://doi.org/https://doi.org/10.1111/1467-9868.00282} {\bibfield
  {journal} {\bibinfo  {journal} {Journal of the Royal Statistical Society.
  Series B, Statistical Methodology}\ }\textbf {\bibinfo {volume} {63}},\
  \bibinfo {pages} {167--241} (\bibinfo {year} {2001})}\BibitemShut {NoStop}%
\bibitem [{\citenamefont {Benth}\ and\ \citenamefont {\v{S}altyt\.{e}
  Benth}(2005)}]{benth05}%
  \BibitemOpen
  \bibfield  {author} {\bibinfo {author} {\bibnamefont {Benth}, \bibfnamefont
  {F.~E.}}and\ \bibinfo {author} {\bibnamefont {\v{S}altyt\.{e} Benth},
  \bibfnamefont {J.}},\ }\bibfield  {title} {\enquote {\bibinfo {title}
  {Stochastic modelling of temperature variations with a view towards weather
  derivatives},}\ }\href
  {https://doi.org/https://doi.org/10.1080/1350486042000271638} {\bibfield
  {journal} {\bibinfo  {journal} {Applied Mathematical Finance}\ }\textbf
  {\bibinfo {volume} {12}},\ \bibinfo {pages} {53--85} (\bibinfo {year}
  {2005})}\BibitemShut {NoStop}%
\bibitem [{\citenamefont {Benth}\ and\ \citenamefont {\v{S}altyt\.{e}
  Benth}(2009)}]{benth09}%
  \BibitemOpen
  \bibfield  {author} {\bibinfo {author} {\bibnamefont {Benth}, \bibfnamefont
  {F.~E.}}and\ \bibinfo {author} {\bibnamefont {\v{S}altyt\.{e} Benth},
  \bibfnamefont {J.}},\ }\bibfield  {title} {\enquote {\bibinfo {title}
  {Dynamic pricing of wind futures},}\ }\href
  {https://doi.org/https://doi.org/10.1016/j.eneco.2008.09.009} {\bibfield
  {journal} {\bibinfo  {journal} {Energy economics}\ }\textbf {\bibinfo
  {volume} {31}},\ \bibinfo {pages} {16--24} (\bibinfo {year}
  {2009})}\BibitemShut {NoStop}%
\bibitem [{\citenamefont {Benth}\ and\ \citenamefont {\v{S}altyt\.{e}
  Benth}(2013)}]{benthbook13}%
  \BibitemOpen
  \bibfield  {author} {\bibinfo {author} {\bibnamefont {Benth}, \bibfnamefont
  {F.~E.}}and\ \bibinfo {author} {\bibnamefont {\v{S}altyt\.{e} Benth},
  \bibfnamefont {J.}},\ }\href@noop {} {\emph {\bibinfo {title} {Modeling And
  Pricing In Financial Markets For Weather Derivatives}}}\ (\bibinfo
  {publisher} {World Scientific},\ \bibinfo {year} {2013})\BibitemShut
  {NoStop}%
\bibitem [{\citenamefont {Benth}, \citenamefont {\v{S}altyt\.{e} Benth},\ and\
  \citenamefont {Koekebakker}(2008)}]{benth2008}%
  \BibitemOpen
  \bibfield  {author} {\bibinfo {author} {\bibnamefont {Benth}, \bibfnamefont
  {F.~E.}}, \bibinfo {author} {\bibnamefont {\v{S}altyt\.{e} Benth},
  \bibfnamefont {J.}}, and\ \bibinfo {author} {\bibnamefont {Koekebakker},
  \bibfnamefont {S.}},\ }\href@noop {} {\emph {\bibinfo {title} {Stochastic
  modelling of electricity and related markets}}}\ (\bibinfo  {publisher}
  {(Vol. 11, Advanced series on statistical science \& applied probability).
  Singapore: World Scientific Publishing Pte.},\ \bibinfo {year}
  {2008})\BibitemShut {NoStop}%
\bibitem [{\citenamefont {Benth}\ and\ \citenamefont
  {Khedher}(2015)}]{benth15}%
  \BibitemOpen
  \bibfield  {author} {\bibinfo {author} {\bibnamefont {Benth}, \bibfnamefont
  {F.~E.}}and\ \bibinfo {author} {\bibnamefont {Khedher}, \bibfnamefont {A.}},\
  }\enquote {\bibinfo {title} {The fascination of probability, statistics and
  their applications},}\ \ (\bibinfo  {publisher} {Cham: Springer International
  Publishing},\ \bibinfo {year} {2015})\ Chap.\ \bibinfo {chapter} {Weak
  Stationarity of Ornstein-Uhlenbeck Processes with Stochastic Speed of Mean
  Reversion}, pp.\ \bibinfo {pages} {153--189}\BibitemShut {NoStop}%
\bibitem [{\citenamefont {Benth}\ \emph {et~al.}(2014)\citenamefont {Benth},
  \citenamefont {Kl{\"u}ppelberg}, \citenamefont {M{\"u}ller},\ and\
  \citenamefont {Vos}}]{benth14}%
  \BibitemOpen
  \bibfield  {author} {\bibinfo {author} {\bibnamefont {Benth}, \bibfnamefont
  {F.~E.}}, \bibinfo {author} {\bibnamefont {Kl{\"u}ppelberg}, \bibfnamefont
  {C.}}, \bibinfo {author} {\bibnamefont {M{\"u}ller}, \bibfnamefont {G.}},
  and\ \bibinfo {author} {\bibnamefont {Vos}, \bibfnamefont {L.}},\ }\bibfield
  {title} {\enquote {\bibinfo {title} {Futures pricing in electricity markets
  based on stable {CARMA} spot models},}\ }\href
  {https://doi.org/https://doi.org/10.1016/j.eneco.2014.03.020} {\bibfield
  {journal} {\bibinfo  {journal} {Energy economics}\ }\textbf {\bibinfo
  {volume} {44}},\ \bibinfo {pages} {392--406} (\bibinfo {year}
  {2014})}\BibitemShut {NoStop}%
\bibitem [{\citenamefont {Benth}\ and\ \citenamefont
  {Taib}(2012)}]{benth_speed_mean_reversion_note}%
  \BibitemOpen
  \bibfield  {author} {\bibinfo {author} {\bibnamefont {Benth}, \bibfnamefont
  {F.~E.}}and\ \bibinfo {author} {\bibnamefont {Taib}, \bibfnamefont {C.~M.
  I.~C.}},\ }\href {http://www.iot.ntnu.no/ef2012/files/papers/37.pdf}
  {\enquote {\bibinfo {title} {On the speed towards the mean for carma
  processes with applications to energy markets},}\ } (\bibinfo {year}
  {2012})\BibitemShut {NoStop}%
\bibitem [{\citenamefont {Benth}\ and\ \citenamefont {Taib}(2013)}]{benth13}%
  \BibitemOpen
  \bibfield  {author} {\bibinfo {author} {\bibnamefont {Benth}, \bibfnamefont
  {F.~E.}}and\ \bibinfo {author} {\bibnamefont {Taib}, \bibfnamefont {C.~M.
  I.~C.}},\ }\bibfield  {title} {\enquote {\bibinfo {title} {On the speed
  towards the mean for continuous time autoregressive moving average processes
  with applications to energy markets},}\ }\href
  {https://doi.org/https://doi.org/10.1016/j.eneco.2013.07.007} {\bibfield
  {journal} {\bibinfo  {journal} {Energy economics}\ }\textbf {\bibinfo
  {volume} {40}},\ \bibinfo {pages} {259--268} (\bibinfo {year}
  {2013})}\BibitemShut {NoStop}%
\bibitem [{\citenamefont {Berrisford}\ \emph {et~al.}(2011)\citenamefont
  {Berrisford}, \citenamefont {Dee}, \citenamefont {Poli}, \citenamefont
  {Brugge}, \citenamefont {Fielding}, \citenamefont {Fuentes}, \citenamefont
  {K{\r a}llberg}, \citenamefont {Kobayashi}, \citenamefont {Uppala},\ and\
  \citenamefont {Simmons}}]{ecmwf_stratospheric_temp}%
  \BibitemOpen
  \bibfield  {author} {\bibinfo {author} {\bibnamefont {Berrisford},
  \bibfnamefont {P.}}, \bibinfo {author} {\bibnamefont {Dee}, \bibfnamefont
  {D.~P.}}, \bibinfo {author} {\bibnamefont {Poli}, \bibfnamefont {P.}},
  \bibinfo {author} {\bibnamefont {Brugge}, \bibfnamefont {R.}}, \bibinfo
  {author} {\bibnamefont {Fielding}, \bibfnamefont {M.}}, \bibinfo {author}
  {\bibnamefont {Fuentes}, \bibfnamefont {M.}}, \bibinfo {author} {\bibnamefont
  {K{\r a}llberg}, \bibfnamefont {P.~W.}}, \bibinfo {author} {\bibnamefont
  {Kobayashi}, \bibfnamefont {S.}}, \bibinfo {author} {\bibnamefont {Uppala},
  \bibfnamefont {S.}}, and\ \bibinfo {author} {\bibnamefont {Simmons},
  \bibfnamefont {A.}},\ }\bibfield  {title} {\enquote {\bibinfo {title} {The
  {ERA}-{Interim} archive {Version} 2.0},}\ }\href
  {https://www.ecmwf.int/node/8174} {\ ,\ \bibinfo {pages} {1--23} (\bibinfo
  {year} {2011})}\BibitemShut {NoStop}%
\bibitem [{\citenamefont {Brockwell}(2001)}]{brockwell2001levy}%
  \BibitemOpen
  \bibfield  {author} {\bibinfo {author} {\bibnamefont {Brockwell},
  \bibfnamefont {P.~J.}},\ }\bibfield  {title} {\enquote {\bibinfo {title}
  {{L{\'e}vy}-driven {CARMA} processes},}\ }\href
  {https://doi.org/https://doi.org/10.1023/A:1017972605872} {\bibfield
  {journal} {\bibinfo  {journal} {Annals of the Institute of Statistical
  Mathematics}\ }\textbf {\bibinfo {volume} {53}},\ \bibinfo {pages} {113--124}
  (\bibinfo {year} {2001})}\BibitemShut {NoStop}%
\bibitem [{\citenamefont {Brockwell}(2004)}]{brockwell04}%
  \BibitemOpen
  \bibfield  {author} {\bibinfo {author} {\bibnamefont {Brockwell},
  \bibfnamefont {P.~J.}},\ }\bibfield  {title} {\enquote {\bibinfo {title}
  {Representations of continuous-time {ARMA} processes},}\ }\href
  {https://doi.org/https://doi.org/10.1239/jap/1082552212} {\bibfield
  {journal} {\bibinfo  {journal} {Journal of Applied Probability}\ }\textbf
  {\bibinfo {volume} {41}},\ \bibinfo {pages} {375--382} (\bibinfo {year}
  {2004})}\BibitemShut {NoStop}%
\bibitem [{\citenamefont {Brockwell}(2014)}]{brockwell2014recent}%
  \BibitemOpen
  \bibfield  {author} {\bibinfo {author} {\bibnamefont {Brockwell},
  \bibfnamefont {P.~J.}},\ }\bibfield  {title} {\enquote {\bibinfo {title}
  {Recent results in the theory and applications of {CARMA} processes},}\
  }\href {https://doi.org/https://doi.org/10.1007/s10463-014-0468-7} {\bibfield
   {journal} {\bibinfo  {journal} {Annals of the Institute of Statistical
  Mathematics}\ }\textbf {\bibinfo {volume} {66}},\ \bibinfo {pages} {647--685}
  (\bibinfo {year} {2014})}\BibitemShut {NoStop}%
\bibitem [{\citenamefont {Brockwell}\ and\ \citenamefont
  {Lindner}(2015)}]{brockwell15}%
  \BibitemOpen
  \bibfield  {author} {\bibinfo {author} {\bibnamefont {Brockwell},
  \bibfnamefont {P.~J.}}and\ \bibinfo {author} {\bibnamefont {Lindner},
  \bibfnamefont {A.}},\ }\bibfield  {title} {\enquote {\bibinfo {title}
  {Prediction of {Lévy}-driven {CARMA} processes},}\ }\href
  {https://doi.org/https://doi.org/10.1016/j.jeconom.2015.03.021} {\bibfield
  {journal} {\bibinfo  {journal} {Journal of Econometrics}\ }\textbf {\bibinfo
  {volume} {189}},\ \bibinfo {pages} {263--271} (\bibinfo {year}
  {2015})}\BibitemShut {NoStop}%
\bibitem [{\citenamefont {Butler}\ \emph {et~al.}(2019)\citenamefont {Butler},
  \citenamefont {Charlton-Perez}, \citenamefont {Domeisen}, \citenamefont
  {Garfinkel}, \citenamefont {Gerber}, \citenamefont {Hitchcock}, \citenamefont
  {Karpechko}, \citenamefont {Maycock}, \citenamefont {Sigmond}, \citenamefont
  {Simpson},\ and\ \citenamefont {Son}}]{butler2019chapter}%
  \BibitemOpen
  \bibfield  {author} {\bibinfo {author} {\bibnamefont {Butler}, \bibfnamefont
  {A.}}, \bibinfo {author} {\bibnamefont {Charlton-Perez}, \bibfnamefont {A.}},
  \bibinfo {author} {\bibnamefont {Domeisen}, \bibfnamefont {D.~I.}}, \bibinfo
  {author} {\bibnamefont {Garfinkel}, \bibfnamefont {C.}}, \bibinfo {author}
  {\bibnamefont {Gerber}, \bibfnamefont {E.~P.}}, \bibinfo {author}
  {\bibnamefont {Hitchcock}, \bibfnamefont {P.}}, \bibinfo {author}
  {\bibnamefont {Karpechko}, \bibfnamefont {A.~Y.}}, \bibinfo {author}
  {\bibnamefont {Maycock}, \bibfnamefont {A.~C.}}, \bibinfo {author}
  {\bibnamefont {Sigmond}, \bibfnamefont {M.}}, \bibinfo {author} {\bibnamefont
  {Simpson}, \bibfnamefont {I.}}, and\ \bibinfo {author} {\bibnamefont {Son},
  \bibfnamefont {S.-W.}},\ }\enquote {\bibinfo {title} {Sub-seasonal to
  seasonal prediction},}\ \ (\bibinfo  {publisher} {Elsevier},\ \bibinfo {year}
  {2019})\ Chap.\ \bibinfo {chapter} {Sub-seasonal Predictability and the
  Stratosphere}, pp.\ \bibinfo {pages} {223--241}\BibitemShut {NoStop}%
\bibitem [{\citenamefont {Butler}\ \emph {et~al.}(2015)\citenamefont {Butler},
  \citenamefont {Seidel}, \citenamefont {Hardiman}, \citenamefont {Butchart},
  \citenamefont {Birner},\ and\ \citenamefont {Match}}]{butler2015defining}%
  \BibitemOpen
  \bibfield  {author} {\bibinfo {author} {\bibnamefont {Butler}, \bibfnamefont
  {A.~H.}}, \bibinfo {author} {\bibnamefont {Seidel}, \bibfnamefont {D.~J.}},
  \bibinfo {author} {\bibnamefont {Hardiman}, \bibfnamefont {S.~C.}}, \bibinfo
  {author} {\bibnamefont {Butchart}, \bibfnamefont {N.}}, \bibinfo {author}
  {\bibnamefont {Birner}, \bibfnamefont {T.}}, and\ \bibinfo {author}
  {\bibnamefont {Match}, \bibfnamefont {A.}},\ }\bibfield  {title} {\enquote
  {\bibinfo {title} {Defining sudden stratospheric warmings},}\ }\href
  {https://doi.org/https://doi.org/10.1175/BAMS-D-13-00173.1} {\bibfield
  {journal} {\bibinfo  {journal} {Bulletin of the American Meteorological
  Society}\ }\textbf {\bibinfo {volume} {96}},\ \bibinfo {pages} {1913--1928}
  (\bibinfo {year} {2015})}\BibitemShut {NoStop}%
\bibitem [{\citenamefont {Clewlow}\ and\ \citenamefont
  {Strickland}(2000)}]{clewlow00}%
  \BibitemOpen
  \bibfield  {author} {\bibinfo {author} {\bibnamefont {Clewlow}, \bibfnamefont
  {L.}}and\ \bibinfo {author} {\bibnamefont {Strickland}, \bibfnamefont {C.}},\
  }\href@noop {} {\emph {\bibinfo {title} {Energy Derivatives-Pricing and Risk
  Management}}}\ (\bibinfo  {publisher} {Lacima Publishers},\ \bibinfo {year}
  {2000})\BibitemShut {NoStop}%
\bibitem [{\citenamefont {Cnossen}, \citenamefont {La\v{s}tovi\v{c}ka},\ and\
  \citenamefont {Emmert}(2015)}]{cnossen2015}%
  \BibitemOpen
  \bibfield  {author} {\bibinfo {author} {\bibnamefont {Cnossen}, \bibfnamefont
  {I.}}, \bibinfo {author} {\bibnamefont {La\v{s}tovi\v{c}ka}, \bibfnamefont
  {J.}}, and\ \bibinfo {author} {\bibnamefont {Emmert}, \bibfnamefont
  {J.~T.}},\ }\bibfield  {title} {\enquote {\bibinfo {title} {Introduction to
  special issue on “long‐term changes and trends in the stratosphere,
  mesosphere, thermosphere and ionosphere”},}\ }\href
  {https://doi.org/https://doi.org/10.1002/2015JD024133} {\bibfield  {journal}
  {\bibinfo  {journal} {Journal of geophysical research: Atmospheres}\ }\textbf
  {\bibinfo {volume} {120}},\ \bibinfo {pages} {11401--11403} (\bibinfo {year}
  {2015})}\BibitemShut {NoStop}%
\bibitem [{\citenamefont {Danilov}\ and\ \citenamefont
  {Konstantinova}(2020)}]{danilov20}%
  \BibitemOpen
  \bibfield  {author} {\bibinfo {author} {\bibnamefont {Danilov}, \bibfnamefont
  {A.~D.}}and\ \bibinfo {author} {\bibnamefont {Konstantinova}, \bibfnamefont
  {A.~V.}},\ }\bibfield  {title} {\enquote {\bibinfo {title} {Long-term
  variations in the parameters of the middle and upper atmosphere and
  ionosphere (review)},}\ }\href
  {https://doi.org/https://doi.org/10.1134/S0016793220040040} {\bibfield
  {journal} {\bibinfo  {journal} {Geomagnetism and Aeronomy}\ }\textbf
  {\bibinfo {volume} {60}},\ \bibinfo {pages} {397--420} (\bibinfo {year}
  {2020})}\BibitemShut {NoStop}%
\bibitem [{\citenamefont {Dee}\ \emph {et~al.}(2011)\citenamefont {Dee},
  \citenamefont {Uppala}, \citenamefont {Simmons}, \citenamefont {Berrisford},
  \citenamefont {Poli}, \citenamefont {Kobayashi}, \citenamefont {Andrae},
  \citenamefont {Balmaseda}, \citenamefont {Balsamo}, \citenamefont {Bauer},
  \citenamefont {Bechtold}, \citenamefont {Beljaars}, \citenamefont {van~de
  Berg}, \citenamefont {Bidlot}, \citenamefont {Bormann}, \citenamefont
  {Delsol}, \citenamefont {Dragani}, \citenamefont {Fuentes}, \citenamefont
  {Geer}, \citenamefont {Haimberger}, \citenamefont {Healy}, \citenamefont
  {Hersbach}, \citenamefont {Hólm}, \citenamefont {Isaksen}, \citenamefont
  {Kållberg}, \citenamefont {K{\"o}hler}, \citenamefont {Matricardi},
  \citenamefont {McNally}, \citenamefont {Monge‐Sanz}, \citenamefont
  {Morcrette}, \citenamefont {Park}, \citenamefont {Peubey}, \citenamefont
  {de~Rosnay}, \citenamefont {Tavolato}, \citenamefont {Thépaut},\ and\
  \citenamefont {Vitart}}]{dee2011}%
  \BibitemOpen
  \bibfield  {author} {\bibinfo {author} {\bibnamefont {Dee}, \bibfnamefont
  {D.~P.}}, \bibinfo {author} {\bibnamefont {Uppala}, \bibfnamefont {S.~M.}},
  \bibinfo {author} {\bibnamefont {Simmons}, \bibfnamefont {A.~J.}}, \bibinfo
  {author} {\bibnamefont {Berrisford}, \bibfnamefont {P.}}, \bibinfo {author}
  {\bibnamefont {Poli}, \bibfnamefont {P.}}, \bibinfo {author} {\bibnamefont
  {Kobayashi}, \bibfnamefont {S.}}, \bibinfo {author} {\bibnamefont {Andrae},
  \bibfnamefont {U.}}, \bibinfo {author} {\bibnamefont {Balmaseda},
  \bibfnamefont {M.~A.}}, \bibinfo {author} {\bibnamefont {Balsamo},
  \bibfnamefont {G.}}, \bibinfo {author} {\bibnamefont {Bauer}, \bibfnamefont
  {P.}}, \bibinfo {author} {\bibnamefont {Bechtold}, \bibfnamefont {P.}},
  \bibinfo {author} {\bibnamefont {Beljaars}, \bibfnamefont {A.~C.~M.}},
  \bibinfo {author} {\bibnamefont {van~de Berg}, \bibfnamefont {L.}}, \bibinfo
  {author} {\bibnamefont {Bidlot}, \bibfnamefont {J.}}, \bibinfo {author}
  {\bibnamefont {Bormann}, \bibfnamefont {N.}}, \bibinfo {author} {\bibnamefont
  {Delsol}, \bibfnamefont {C.}}, \bibinfo {author} {\bibnamefont {Dragani},
  \bibfnamefont {R.}}, \bibinfo {author} {\bibnamefont {Fuentes}, \bibfnamefont
  {M.}}, \bibinfo {author} {\bibnamefont {Geer}, \bibfnamefont {A.~J.}},
  \bibinfo {author} {\bibnamefont {Haimberger}, \bibfnamefont {L.}}, \bibinfo
  {author} {\bibnamefont {Healy}, \bibfnamefont {S.~B.}}, \bibinfo {author}
  {\bibnamefont {Hersbach}, \bibfnamefont {H.}}, \bibinfo {author}
  {\bibnamefont {Hólm}, \bibfnamefont {E.~V.}}, \bibinfo {author}
  {\bibnamefont {Isaksen}, \bibfnamefont {L.}}, \bibinfo {author} {\bibnamefont
  {Kållberg}, \bibfnamefont {P.}}, \bibinfo {author} {\bibnamefont
  {K{\"o}hler}, \bibfnamefont {M.}}, \bibinfo {author} {\bibnamefont
  {Matricardi}, \bibfnamefont {M.}}, \bibinfo {author} {\bibnamefont {McNally},
  \bibfnamefont {A.~P.}}, \bibinfo {author} {\bibnamefont {Monge‐Sanz},
  \bibfnamefont {B.~M.}}, \bibinfo {author} {\bibnamefont {Morcrette},
  \bibfnamefont {J.~J.}}, \bibinfo {author} {\bibnamefont {Park}, \bibfnamefont
  {B.~K.}}, \bibinfo {author} {\bibnamefont {Peubey}, \bibfnamefont {C.}},
  \bibinfo {author} {\bibnamefont {de~Rosnay}, \bibfnamefont {P.}}, \bibinfo
  {author} {\bibnamefont {Tavolato}, \bibfnamefont {C.}}, \bibinfo {author}
  {\bibnamefont {Thépaut}, \bibfnamefont {J.~N.}}, and\ \bibinfo {author}
  {\bibnamefont {Vitart}, \bibfnamefont {F.}},\ }\bibfield  {title} {\enquote
  {\bibinfo {title} {The {ERA‐Interim} reanalysis: configuration and
  performance of the data assimilation system},}\ }\href
  {https://doi.org/https://doi.org/10.1002/qj.828} {\bibfield  {journal}
  {\bibinfo  {journal} {Quarterly Journal of the Royal Meteorological Society}\
  }\textbf {\bibinfo {volume} {137}},\ \bibinfo {pages} {553--597} (\bibinfo
  {year} {2011})}\BibitemShut {NoStop}%
\bibitem [{\citenamefont {Fu}, \citenamefont {Solomon},\ and\ \citenamefont
  {Lin}(2010)}]{fu10}%
  \BibitemOpen
  \bibfield  {author} {\bibinfo {author} {\bibnamefont {Fu}, \bibfnamefont
  {Q.}}, \bibinfo {author} {\bibnamefont {Solomon}, \bibfnamefont {S.}}, and\
  \bibinfo {author} {\bibnamefont {Lin}, \bibfnamefont {P.}},\ }\bibfield
  {title} {\enquote {\bibinfo {title} {On the seasonal dependence of tropical
  lower-stratospheric temperature trends},}\ }\href
  {https://doi.org/https://doi.org/10.5194/acp-10-2643-2010} {\bibfield
  {journal} {\bibinfo  {journal} {Atmospheric Chemistry and Physics}\ }\textbf
  {\bibinfo {volume} {10}},\ \bibinfo {pages} {2643--2653} (\bibinfo {year}
  {2010})}\BibitemShut {NoStop}%
\bibitem [{\citenamefont {Haynes}(2005)}]{haynes2005}%
  \BibitemOpen
  \bibfield  {author} {\bibinfo {author} {\bibnamefont {Haynes}, \bibfnamefont
  {P.}},\ }\bibfield  {title} {\enquote {\bibinfo {title} {Stratospheric
  dynamics},}\ }\href
  {https://doi.org/https://doi-org.ezproxy.uio.no/10.1146/annurev.fluid.37.061903.175710}
  {\bibfield  {journal} {\bibinfo  {journal} {Annual Review of Fluid
  Mechanics}\ }\textbf {\bibinfo {volume} {37}},\ \bibinfo {pages} {263--293}
  (\bibinfo {year} {2005})}\BibitemShut {NoStop}%
\bibitem [{\citenamefont {Iacus}(2008)}]{iacus08}%
  \BibitemOpen
  \bibfield  {author} {\bibinfo {author} {\bibnamefont {Iacus}, \bibfnamefont
  {S.~M.}},\ }\href@noop {} {\emph {\bibinfo {title} {Simulation and Inference
  for Stochastic Differential Equations: With R Examples}}}\ (\bibinfo
  {publisher} {Springer New York},\ \bibinfo {year} {2008})\BibitemShut
  {NoStop}%
\bibitem [{\citenamefont {{Intergovernmental Panel on Climate
  Change}}(2014)}]{climateChange2013}%
  \BibitemOpen
  \bibfield  {author} {\bibinfo {author} {\bibnamefont {{Intergovernmental
  Panel on Climate Change}},},\ }\href
  {https://doi.org/10.1017/CBO9781107415324} {\emph {\bibinfo {title} {Climate
  Change 2013 – The Physical Science Basis: Working Group I Contribution to
  the Fifth Assessment Report of the Intergovernmental Panel on Climate
  Change}}}\ (\bibinfo  {publisher} {Cambridge University Press},\ \bibinfo
  {year} {2014})\BibitemShut {NoStop}%
\bibitem [{\citenamefont {Jones}\ and\ \citenamefont
  {Wigley}(1990)}]{jones1990globalwarming}%
  \BibitemOpen
  \bibfield  {author} {\bibinfo {author} {\bibnamefont {Jones}, \bibfnamefont
  {P.~D.}}and\ \bibinfo {author} {\bibnamefont {Wigley}, \bibfnamefont
  {T.~M.~L.}},\ }\bibfield  {title} {\enquote {\bibinfo {title} {Global warming
  trends},}\ }\href
  {https://doi.org/https://doi.org/10.1038/scientificamerican0890-84}
  {\bibfield  {journal} {\bibinfo  {journal} {Scientific American}\ }\textbf
  {\bibinfo {volume} {263}},\ \bibinfo {pages} {84--91} (\bibinfo {year}
  {1990})}\BibitemShut {NoStop}%
\bibitem [{\citenamefont {Karpechko}, \citenamefont {Tummon},\ and\
  \citenamefont {{WMO Secretariat}}(2016)}]{KarpechkoWMO2016}%
  \BibitemOpen
  \bibfield  {author} {\bibinfo {author} {\bibnamefont {Karpechko},
  \bibfnamefont {A.}}, \bibinfo {author} {\bibnamefont {Tummon}, \bibfnamefont
  {F.}}, and\ \bibinfo {author} {\bibnamefont {{WMO Secretariat}},},\
  }\bibfield  {title} {\enquote {\bibinfo {title} {Climate predictability in
  the stratosphere},}\ }\href@noop {} {\bibfield  {journal} {\bibinfo
  {journal} {WMO Bulletin}\ }\textbf {\bibinfo {volume} {65}} (\bibinfo {year}
  {2016})}\BibitemShut {NoStop}%
\bibitem [{\citenamefont {Kitchin}(2013)}]{sigmoid}%
  \BibitemOpen
  \bibfield  {author} {\bibinfo {author} {\bibnamefont {Kitchin}, \bibfnamefont
  {J.}},\ }\href
  {https://kitchingroup.cheme.cmu.edu/blog/2013/01/31/Smooth-transitions-between-discontinuous-functions/}
  {\enquote {\bibinfo {title} {Smooth transitions between discontinuous
  functions},}\ } (\bibinfo {year} {2013})\BibitemShut {NoStop}%
\bibitem [{\citenamefont {Levendis}(2018)}]{Levendis18}%
  \BibitemOpen
  \bibfield  {author} {\bibinfo {author} {\bibnamefont {Levendis},
  \bibfnamefont {J.~D.}},\ }\href@noop {} {\emph {\bibinfo {title} {Time Series
  Econometrics: {Learning} Through Replication}}}\ (\bibinfo  {publisher}
  {Springer International Publishing: Imprint: Springer},\ \bibinfo {year}
  {2018})\BibitemShut {NoStop}%
\bibitem [{\citenamefont {McCormack}\ and\ \citenamefont
  {Hood}(1996)}]{mccormack96}%
  \BibitemOpen
  \bibfield  {author} {\bibinfo {author} {\bibnamefont {McCormack},
  \bibfnamefont {J.~P.}}and\ \bibinfo {author} {\bibnamefont {Hood},
  \bibfnamefont {L.~L.}},\ }\bibfield  {title} {\enquote {\bibinfo {title}
  {Apparent solar cycle variations of upper stratospheric ozone and
  temperature: {Latitude} and seasonal dependences},}\ }\href
  {https://doi.org/https://doi.org/10.1029/96JD01817} {\bibfield  {journal}
  {\bibinfo  {journal} {Journal of Geophysical Research: Atmospheres}\ }\textbf
  {\bibinfo {volume} {101}},\ \bibinfo {pages} {20933--20944} (\bibinfo {year}
  {1996})}\BibitemShut {NoStop}%
\bibitem [{\citenamefont {Pedatella}\ \emph {et~al.}(2018)\citenamefont
  {Pedatella}, \citenamefont {Chau}, \citenamefont {Schmidt}, \citenamefont
  {Goncharenko}, \citenamefont {Stolle}, \citenamefont {Hocke}, \citenamefont
  {Harvey}, \citenamefont {Funke},\ and\ \citenamefont
  {Siddiqui}}]{Pedatella2018}%
  \BibitemOpen
  \bibfield  {author} {\bibinfo {author} {\bibnamefont {Pedatella},
  \bibfnamefont {N.}}, \bibinfo {author} {\bibnamefont {Chau}, \bibfnamefont
  {J.}}, \bibinfo {author} {\bibnamefont {Schmidt}, \bibfnamefont {H.}},
  \bibinfo {author} {\bibnamefont {Goncharenko}, \bibfnamefont {L.}}, \bibinfo
  {author} {\bibnamefont {Stolle}, \bibfnamefont {C.}}, \bibinfo {author}
  {\bibnamefont {Hocke}, \bibfnamefont {K.}}, \bibinfo {author} {\bibnamefont
  {Harvey}, \bibfnamefont {V.}}, \bibinfo {author} {\bibnamefont {Funke},
  \bibfnamefont {B.}}, and\ \bibinfo {author} {\bibnamefont {Siddiqui},
  \bibfnamefont {T.}},\ }\bibfield  {title} {\enquote {\bibinfo {title} {How
  sudden stratospheric warming affects the whole atmosphere},}\ }\href
  {https://doi.org/https://doi.org/10.1029/2018EO092441} {\bibfield  {journal}
  {\bibinfo  {journal} {Eos}\ }\textbf {\bibinfo {volume} {99}} (\bibinfo
  {year} {2018}),\ https://doi.org/10.1029/2018EO092441}\BibitemShut {NoStop}%
\bibitem [{\citenamefont {Steiner}\ \emph {et~al.}(2020)\citenamefont
  {Steiner}, \citenamefont {Ladst{\"a}dter}, \citenamefont {Randel},
  \citenamefont {Maycock}, \citenamefont {Fu}, \citenamefont {Claud},
  \citenamefont {Gleisner}, \citenamefont {Haimberger}, \citenamefont {Ho},
  \citenamefont {Keckhut}, \citenamefont {Leblanc}, \citenamefont {Mears},
  \citenamefont {Polvani}, \citenamefont {Santer}, \citenamefont {Schmidt},
  \citenamefont {Sofieva}, \citenamefont {Wing}, \citenamefont {Zou},\ and\
  \citenamefont {Cardon}}]{steiner20}%
  \BibitemOpen
  \bibfield  {author} {\bibinfo {author} {\bibnamefont {Steiner}, \bibfnamefont
  {A.~K.}}, \bibinfo {author} {\bibnamefont {Ladst{\"a}dter}, \bibfnamefont
  {F.}}, \bibinfo {author} {\bibnamefont {Randel}, \bibfnamefont {W.~J.}},
  \bibinfo {author} {\bibnamefont {Maycock}, \bibfnamefont {A.~C.}}, \bibinfo
  {author} {\bibnamefont {Fu}, \bibfnamefont {Q.}}, \bibinfo {author}
  {\bibnamefont {Claud}, \bibfnamefont {C.}}, \bibinfo {author} {\bibnamefont
  {Gleisner}, \bibfnamefont {H.}}, \bibinfo {author} {\bibnamefont
  {Haimberger}, \bibfnamefont {L.}}, \bibinfo {author} {\bibnamefont {Ho},
  \bibfnamefont {S.~P.}}, \bibinfo {author} {\bibnamefont {Keckhut},
  \bibfnamefont {P.}}, \bibinfo {author} {\bibnamefont {Leblanc}, \bibfnamefont
  {T.}}, \bibinfo {author} {\bibnamefont {Mears}, \bibfnamefont {C.}}, \bibinfo
  {author} {\bibnamefont {Polvani}, \bibfnamefont {L.~M.}}, \bibinfo {author}
  {\bibnamefont {Santer}, \bibfnamefont {B.~D.}}, \bibinfo {author}
  {\bibnamefont {Schmidt}, \bibfnamefont {T.}}, \bibinfo {author} {\bibnamefont
  {Sofieva}, \bibfnamefont {V.}}, \bibinfo {author} {\bibnamefont {Wing},
  \bibfnamefont {R.}}, \bibinfo {author} {\bibnamefont {Zou}, \bibfnamefont
  {C.~Z.}}, and\ \bibinfo {author} {\bibnamefont {Cardon}, \bibfnamefont
  {C.}},\ }\bibfield  {title} {\enquote {\bibinfo {title} {Observed temperature
  changes in the troposphere and stratosphere from 1979 to 2018},}\ }\href
  {https://doi.org/https://doi.org/10.1175/JCLI-D-19-0998.1} {\bibfield
  {journal} {\bibinfo  {journal} {Journal of climate}\ }\textbf {\bibinfo
  {volume} {33}},\ \bibinfo {pages} {8165--8194} (\bibinfo {year}
  {2020})}\BibitemShut {NoStop}%
\bibitem [{\citenamefont {Sévellec}\ and\ \citenamefont
  {Drijfhout}(2018)}]{sevellec2018globalwarming}%
  \BibitemOpen
  \bibfield  {author} {\bibinfo {author} {\bibnamefont {Sévellec},
  \bibfnamefont {F.}}and\ \bibinfo {author} {\bibnamefont {Drijfhout},
  \bibfnamefont {S.~S.}},\ }\bibfield  {title} {\enquote {\bibinfo {title} {A
  novel probabilistic forecast system predicting anomalously warm 2018-2022
  reinforcing the long-term global warming trend},}\ }\href
  {https://doi.org/https://doi.org/10.1038/s41467-018-05442-8} {\bibfield
  {journal} {\bibinfo  {journal} {Nature Communications}\ }\textbf {\bibinfo
  {volume} {9}},\ \bibinfo {pages} {3024--3024} (\bibinfo {year}
  {2018})}\BibitemShut {NoStop}%
\bibitem [{\citenamefont {Vallis}(2017)}]{vallis2017}%
  \BibitemOpen
  \bibfield  {author} {\bibinfo {author} {\bibnamefont {Vallis}, \bibfnamefont
  {G.~K.}},\ }\enquote {\bibinfo {title} {Atmospheric and oceanic fluid
  dynamics},}\ \ (\bibinfo  {publisher} {Cambridge University Press},\ \bibinfo
  {year} {2017})\ Chap.\ \bibinfo {chapter} {The Stratosphere}, pp.\ \bibinfo
  {pages} {627--671}\BibitemShut {NoStop}%
\bibitem [{\citenamefont {Zapranis}\ and\ \citenamefont
  {Alexandridis}(2008)}]{zapranis08}%
  \BibitemOpen
  \bibfield  {author} {\bibinfo {author} {\bibnamefont {Zapranis},
  \bibfnamefont {A.}}and\ \bibinfo {author} {\bibnamefont {Alexandridis},
  \bibfnamefont {A.}},\ }\bibfield  {title} {\enquote {\bibinfo {title}
  {Modelling the temperature time-dependent speed of mean reversion in the
  context of weather derivatives pricing},}\ }\href
  {https://doi.org/https://doi.org/10.1080/13504860802006065} {\bibfield
  {journal} {\bibinfo  {journal} {Applied Mathematical Finance}\ }\textbf
  {\bibinfo {volume} {15}},\ \bibinfo {pages} {355--386} (\bibinfo {year}
  {2008})}\BibitemShut {NoStop}%
\end{thebibliography}%

\end{document}